\documentclass[aps,pre,reprint,groupedaddress,longbibliography]{revtex4-1}
\usepackage{amsmath,amssymb,epsfig,bbm,bm,verbatim,MnSymbol,mathtools,amsthm}
\usepackage[colorlinks=true,citecolor=blue]{hyperref}

\begin{document}


\title{Thermodynamic Bounds on Coherent Transport in Periodically Driven Conductors}


\author{Elina Potanina\textsuperscript{1}}
\author{Christian Flindt\textsuperscript{1}}
\author{Michael Moskalets\textsuperscript{2}}
\author{Kay Brandner\textsuperscript{1,3,4,5}}
\affiliation{\textsuperscript{1}Department of Applied Physics,
Aalto University,
00076 Aalto,
Finland\\
\textsuperscript{2}Department of Metal and Semiconductor Physics,
NTU “Kharkiv Polytechnic Institute”,
61002 Kharkiv,
Ukraine\\
\textsuperscript{3}Department of Physics, 
Keio University, 
3-14-1 Hiyoshi, Kohoku-ku, 
Yokohama 223-8522,
Japan\\
\textsuperscript{4}School of Physics and Astronomy,
University of Nottingham,
Nottingham NG7 2RD,
United Kingdom\\
\textsuperscript{5}Centre for the Mathematics and Theoretical Physics of Quantum Non-equilibrium Systems,
University of Nottingham,
Nottingham NG7 2RD, 
United Kingdom
}


\date{\today}

\begin{abstract}
Periodically driven coherent conductors provide a universal platform 
for the development of quantum transport devices. 
Here, we lay down a comprehensive theory to describe the 
thermodynamics of these systems. 
We first focus on moderate thermo-electrical biases and 
low driving frequencies. 
For this linear response regime, we establish generalized
Onsager-Casimir relations and an extended fluctuation-dissipation
theorem. 
Furthermore, we derive a family of thermodynamic bounds proving 
that any local matter or heat current puts a non-trivial lower limit 
on the overall dissipation rate of a coherent transport process. 
These bounds do not depend on system-specific parameters, are robust
against dephasing and involve only experimentally accessible 
quantities. 
They thus provide powerful tools to optimize the performance of 
mesoscopic devices and for thermodynamic inference, as we demonstrate
by working out three specific applications. 
We then show that physically transparent extensions of our 
bounds hold also for strong biases and high frequencies. 
These generalized bounds imply a thermodynamic uncertainty 
relation that fully accounts for quantum effects and periodic driving.
Moreover, they lead to a universal and operationally accessible bound
on entropy production that can be readily used for thermodynamic 
inference and device engineering far from equilibrium.
Connecting a broad variety of topics that range from thermodynamic 
geometry over thermodynamic uncertainty relations to quantum 
engineering, our work provides a unifying thermodynamic theory of
coherent transport that can be tested and utilized with current
technologies. 
\end{abstract}


\maketitle

\vbadness=10000
\hbadness=10000

\newtheorem{lemma}{Lemma}

\newcommand{\df}{\mathbf{V}}
\renewcommand{\a}{\alpha}
\renewcommand{\b}{\beta}
\renewcommand{\d}{\delta}
\newcommand{\g}{\gamma}
\renewcommand{\r}{\rho}
\newcommand{\w}{\omega}
\renewcommand{\l}{V}
\newcommand{\ve}{\varepsilon}
\newcommand{\ce}{\varepsilon}
\newcommand{\F}{F}
\newcommand{\X}{X}
\newcommand{\Z}{Z}
\newcommand{\G}{G}
\newcommand{\A}{\mathcal{A}}
\newcommand{\AS}{\mathcal{S}}
\newcommand{\T}{\mathcal{T}}
\newcommand{\TB}{\mathsf{T}_{\mathbf{B}}}
\newcommand{\TV}{\mathsf{T}_{\df}}
\newcommand{\hw}{\hbar\omega}
\newcommand{\BC}{\boldsymbol{\mathcal{B}}}
\newcommand{\RS}{\mathcal{R}}
\newcommand{\J}{\mathbf{J}}
\newcommand{\LL}{\mathbb{L}}
\newcommand{\FF}{\boldsymbol{\mathcal{F}}}
\renewcommand{\L}{\mathbf{L}}
\newcommand{\tp}{\mathsf{T}}
\newcommand{\En}{E\ix{n}}
\newcommand{\eq}{|_{{{\rm eq}}}}

\newcommand{\Eint}{\int_0^\infty\!\!\! dE}
\renewcommand{\tint}{\int_0^\T\!\!\! dt \;}
\newcommand{\nsum}{\sum\nolimits_n}
\newcommand{\bsum}{\sum\nolimits_\b}
\newcommand{\asum}{\sum\nolimits_\a}
\newcommand{\gsum}{\sum\nolimits_\g}
\newcommand{\absum}{\sum\nolimits_{\a\b}}
\newcommand{\xsum}{\sum\nolimits_x}
\newcommand{\ysum}{\sum\nolimits_y}
\newcommand{\xysum}{\sum\nolimits_{xy}}
\newcommand{\Asum}{\sum\nolimits_A}
\newcommand{\Bsum}{\sum\nolimits_B}
\newcommand{\ABsum}{\sum\nolimits_{AB}}
\newcommand{\nusum}{\sum\nolimits_\nu}

\newcommand{\ix}[1]{\raisebox{-0.7pt}{{{\scriptsize $#1$}}}}
\newcommand{\abs}[1]{\bigl| #1 \bigr|}
\newcommand{\tb}[1]{\textcolor{blue}{#1}}
\newcommand{\negphantom}[1]{\settowidth{\dimen0}{#1}\hspace*{-\dimen0}}

\section{Introduction}
Transport is a thermodynamic process, where gradients in intensive 
parameters such as chemical potential and temperature drive currents
of extensive quantities like matter and energy. 
In macroscopic systems at high temperatures, this phenomenon can be 
understood as a result of frequent collisions between classical 
particles, which lead to random but biased changes of their direction 
of motion. 
This mechanism is know as diffusive transport \cite{Plawsky2001}. 
Reducing the temperature of the system increases the mean free path 
that particles can travel between consecutive collisions. 
When this length scale becomes comparable to the dimensions of the
conductor, as occurs in nano-scale structures at millikelvin 
temperatures, coherent transport sets in \cite{Imry2002}. 
This regime is governed by the laws of quantum mechanics and can no
longer be describe in terms of collisions between particles with 
well-defined positions and momenta. 
 
Instead, coherent transport can be seen as arising from the unitary 
propagation of beams of carriers that are emitted and absorbed by 
distant thermal reservoirs and undergo elastic scattering within the
conductor. 
This approach goes back to the pioneering work of Landauer
\cite{Landauer1957a} and has since evolved into a standard theoretical
tool of mesoscopic physics \cite{Imry2002,Lesovik2014}. 
In particular, it has been extended to systems that are subject
to oscillating electromagnetic fields, where the scattering of beams 
is still coherent but no longer elastic, since carriers can exchange 
discrete amounts of energy with the driving fields
\cite{Moskalets2002,Moskalets2004,Moskalets2012,Brandner2020}.

On the experimental side, technological progress has made it possible
to realize and control periodically driven coherent conductors with a
high degree of precision.
Today, these systems provide us with a versatile platform to test the
basic principles of thermodynamics at small-length and energy scales 
and to develop new mesoscopic devices such as parametric quantum 
pumps, which can be used to realize dynamical single-electron sources
\cite{Feve2007a,Bocquillon2012,Crowell2013,Dubois2013,Jullien2014,
Ubbelohde2015,Kataoka2016,Takada2018a,Roussely2018,Johnson2018,
Fletcher2019}
as well as for metrological applications 
\cite{Giblin2012,Pashkin2013,Jehl2013}, 
and adiabatic quantum motors, which may provide motive power to future
nano-machines  \cite{Bode2011,Bustos-Marun2013,Arrachea2016,
Bruch2018a}.

These endeavors will require a powerful theoretical framework to 
describe the thermodynamics of coherent transport in the presence of
both thermo-chemical biases and periodic driving. 
A suitable starting point for such a theory is provided by Onsager's 
irreversible thermodynamics \cite{Onsager1931a,Onsager1931,
Callen1985}. 
The key idea of this approach is to describe irreversible processes in
terms of two types of variables: 
thermodynamic forces, which drive the process, and currents, which 
correspond to the system's response. 
This concept is universal in that it can be applied to macroscopic 
\cite{Onsager1931a,Onsager1931,Callen1985} and mesoscopic 
\cite{Sivan1986,Butcher1990} systems alike.  
Moreover, it can be consistently expanded to include periodic driving.
Specifically, in the context of coherent transport, this 
generalization can be achieved by introducing an additional force, 
which is proportional to the frequency of the applied fields, and an 
additional current, which corresponds to the flux of photons that is
absorbed by the carriers inside the conductor \cite{Ludovico2015b}.  

Here, we show that this framework can further be underpinned by 
rigorous generalizations of two cornerstone results of classical
irreversible thermodynamics: the Onsager-Casimir relations
\cite{Onsager1931a,Onsager1931,Casimir1945,Callen1985}, which 
explain the interdependence between linear-response coefficients as
a consequence of microscopic time-reversal symmetry, and the 
fluctuation-dissipation theorem, which connects these coefficients to
equilibrium current fluctuations \cite{Kubo1966,Kubo1998}. 
Focusing on moderate electric and thermal biases and slowly varying
driving fields, we then set out to derive our first key result,
the relations
\begin{equation}\label{Int}
\sigma\geq\frac{N}{4 K^{xx} (N-1)} (J^x_\a)^2,
\end{equation}
which bound the overall dissipation $\sigma$ caused by a coherent 
transport process in an $N$-terminal conductor in terms of any 
period averaged matter current $J^\r_\a$ or heat current $J^q_\a$. 
The coefficients $K^{\r\r}$ and $K^{qq}$ thereby depend only on the
equilibrium temperature and chemical potential of the conductor. 
These bounds are stronger than the second law, which only requires
$\sigma\geq 0$, and universal in that they do not involve any 
system-specific parameters. 

In the second part of this article, we extend our theory to systems
that are driven far away from equilibrium. 
This endeavor leads to our second key result, the relation
\begin{equation}\label{IInt}
\sigma\geq \sqrt{(P^{\r\r}_{\a\a}+2T_\a/h)^2+2\psi^\ast (J^\r_\a)^2}
	-(P^{\r\r}_{\a\a}+2T_\a/h),
\end{equation}
which makes it possible to bound the total dissipation rate $\sigma$
by measuring the average $J^\r_\a$ and the zero-frequency noise 
$P^{\r\r}_{\a\a}$ of a single matter current along with the 
temperature $T_\a$ of the corresponding reservoir; 
$h$ denotes Planck's constant and $\psi^\ast\simeq 8/9$ is a 
numerical factor. 
Quite remarkably, the bound \eqref{IInt} holds for any coherent 
multi-terminal conductor, arbitrary strong thermo-chemical biases and 
arbitrary fast periodic driving fields.

Covering both thermal and mechanical driving, the relations 
\eqref{Int} and \eqref{IInt} lead to non-trivial bounds on the 
figures of merit of cyclic nano-machines based on coherent conductors.
In this respect, they advance an active line of research, which has
so far mainly focused on steady-state devices and is driven
by two major motivations
\cite{Avron2001,Benenti2011,Brandner2013,Whitney2013,Brandner2013a,
Mazza2014,Whitney2014,Sothmann2014a,Brandner2015,Whitney2014a,
Samuelsson2017,Brandner2017b,Macieszczak2018,Luo2018}.
First, universal bounds on figures such as efficiency or power 
consumption make it possible to quantitatively compare and optimize
different theoretical models of mesoscopic devices. 
Second, such relations can be used in experiments to estimate
quantities that cannot be measured directly, a strategy known as
thermodynamic inference \cite{Seifert2018}. 
As we show by working out three specific applications, our theory 
provides powerful tools for both of these purposes. 
In particular, we provide a detailed analysis of parametric 
quantum pumps, for which we uncover a close connection 
between our approach and the concepts of thermodynamic geometry, a
framework that recently proved very useful for optimizing slowly 
driven quantum thermal machines \cite{Brandner2019c,Abiuso2020,
Miller2019a,Miller2020,Miller2020a}. 

We proceed as follows.
In the next section, we review the essentials of the scattering
approach to coherent transport in periodically driven conductors and
show how it can be furnished with a thermodynamic structure. 
We then work out the linear-response theory for this framework and 
lay down the corresponding generalizations of the Onsager-Casimir
relations and the fluctuation dissipation theorem in 
Sec.~\ref{SecFAR}.
In Sec.~\ref{SecTC}, we derive the key relation \eqref{Int} and 
discuss its range of validity. 
The purpose of Sec.~\ref{SecAdChrPump} is to demonstrate the 
versatile applicability of our linear-response results. 
To this end, we consider three different mesoscopic devices. 
As an introductory example, we analyze a simple model of a quantum 
generator, whereby we prove that our new bounds are tight. 
We then move on to parametric quantum pumps, for which we derive an
explicit optimization principle by connecting our theory to the 
framework of thermodynamic geometry. 
Furthermore, we show how our bounds can be used for thermodynamic
inference. 
Finally, we further illustrate this technique by applying it to 
adiabatic quantum motors. 
In the second major part of this article, 
Sec.~\ref{Sec_FFE}, we show how the bounds \eqref{Int} can be 
extended beyond the limits of the linear-response regime and
derive our second key relation \eqref{IInt}.
We then put these results in context with recent developments on
thermodynamic uncertainty relations. 
To this end, we work out a case study, which proves that our 
bounds on entropy production go significantly beyond earlier results. 
Finally, we discuss the implications of our theory for autonomous 
coherent conductors. 
We summarize our work in Sec.~\ref{SecConclusion}.

\section{Scattering Approach}\label{Sec_SA}

\begin{figure}
\includegraphics[scale=1.05]{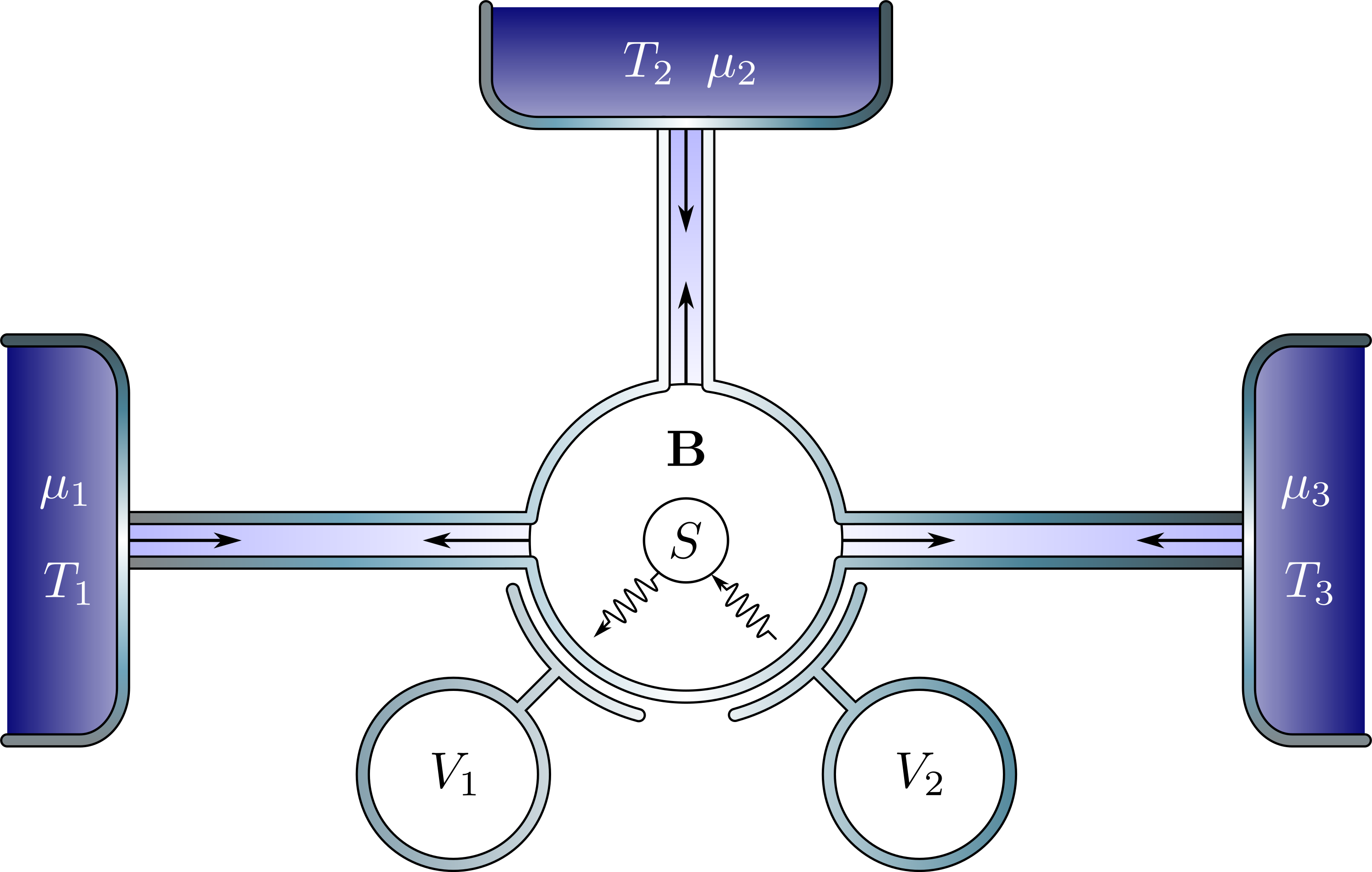}
\caption{Scattering approach to coherent transport. 
The figure shows a three-terminal conductor with two driving fields
as a generic example for a multi-terminal conductor. 
The system consists of a sample $S$, which is subject to a constant 
magnetic field $\mathbf{B}$, and three reservoirs with chemical
potentials $\mu_1,\mu_2,\mu_3$ and temperatures $T_1,T_2,T_3$.
Each reservoir injects a beam of carriers, which consists of a thermal
mixture of plane waves propagating towards the sample. 
At the sample, each incoming beam is scattered into three outgoing 
beams. 
The phase relation between incoming and outgoing carriers is 
preserved, while their energy can change due to the exchange of
photons with the periodic driving fields $V_1$ and $V_2$.
\label{Fig_Scatt}}
\end{figure}

Scattering theory provides an elegant tool to describe coherent 
transport in periodically driven mesoscopic systems. 
In this approach, the conductor is divided into a sample region, where
carriers may be subject to a periodically modulated potential and an
external magnetic field, and a set of ideal leads. 
Each lead is connected to a reservoir, which injects a continuous beam
of thermalized, non-interacting carriers. 
These beams propagates coherently through the system before being 
reabsorbed by the reservoirs see Fig.~\ref{Fig_Scatt}. 
The emerging matter and energy currents are thus determined by the 
inelastic scattering amplitudes of the driven sample and the chemical 
potentials and temperatures of the reservoirs. 
In the following, we provide a brief review of this framework and its
thermodynamic interpretation. 
Further details may be found in 
Refs.~\cite{Moskalets2002,Moskalets2004,Moskalets2012,Brandner2020}.

\subsection{Mean Currents and Fluctuations}
In the Heisenberg picture, the matter and energy currents that enter
the system through the terminal $\a$ correspond to time dependent 
operators $\hat{J}^{\rho}_{\a,t}$ and $\hat{J}^{\ce}_{\a,t}$. 
The mean values and fluctuations of these currents are given by the 
general expressions
\begin{subequations}
\begin{align}
\label{ST_CExpl}
J^u_\a & = \lim_{t\rightarrow\infty} \frac{1}{t}\int_0^t\!\!\! dt'
	\bigl\langle \hat{J}^u_{\a,t'}\bigr\rangle \quad\text{and}\\
\label{ST_NExpl}
P^{uv}_{\a\b} &= \lim_{t\rightarrow\infty}
	\frac{1}{t} \int_0^t\!\!\! dt' \! \int_0^t\!\!\! dt'' 
	\bigl\langle\bigl(\hat{J}^u_{\a,t'}-J^u_\a\bigr)
	\bigl(\hat{J}^v_{\b,t''}-J^v_\b\bigr)\bigr\rangle,
\end{align}
\end{subequations}
where $u,v=\r,\ce$ and angular brackets indicate the average over all
quantum states of the injected carriers. 
Since the carriers are non-interacting, this average can be evaluated 
by treating the incoming beams as ideal Fermionic quantum gases, which
leads to the generalized Landauer-B\"uttiker formula
\begin{align}
\label{TwMeanCurrents}
J^u_\alpha &= \frac{1}{h}\Eint\bsum\nsum
	\Bigl(\xi_E^u \d_{\a\b}\d_{n0}-\xi_{E\ix{n}}^u 
	\abs{S^{\alpha\beta}_{E\ix{n},E}}^2\Bigr)f^\b_E
\end{align}
for the mean currents with $\xi^\rho_E = 1$ and $\xi^\ce_E =E$. 
Here, we have introduced the shorthand notation $E_n \equiv E + n\hw$,
where $\hbar=h/2\pi$ denotes the reduced Planck constant and $\omega
\equiv 2\pi/\T$ the frequency of periodic driving fields acting on 
the sample. 
Thermodynamics enters Eq.~\eqref{TwMeanCurrents} via the Fermi 
functions
\begin{equation}\label{TwFermiF}
f^\alpha_E \equiv \frac{1}{1+\exp[(E-\mu_\alpha)/T_\alpha]},
\end{equation}
where $\mu_\a$ and $T_\a$ are the chemical potential and temperature 
of the reservoir $\a$ and Boltzmann's constant is set to $1$
throughout. 
The properties of the sample are encoded in the Floquet scattering 
amplitudes $S^{\alpha\beta}_{E\ix{n},E}$. 
These objects describe the transmission of an incoming carrier with
energy $E$ from the terminal $\beta$ to the terminal $\alpha$ under 
the absorption of $n$ photons with energy $\hbar\omega$.
Note that, for simplicity, we assume that each lead supports only one
transport channel. 
Furthermore, we use the convention that the photon-counting index runs
over all integers and that the Floquet scattering amplitudes are zero
if one of their energy arguments is negative. 

The current fluctuations \eqref{ST_NExpl} can be evaluated in the same
way as the mean currents. 
The resulting formula involves two contributions, $P^{uv}_{\a\b}= 
D^{uv}_{\a\b} + R^{uv}_{\a\b}$, which are given by
\begin{subequations}\label{TwMEN} 
\begin{align}
\label{TwMENTh}
D^{uv}_{\a\b} &= \frac{1}{h}\Eint \nsum\bigl(
	A^{uv,\a\b}_{n,E} + A^{vu,\b\a}_{n,E}+\delta_{\a\b}B^{uv,\a}_{n,E}
	\bigr),\\
\label{TwMENSh}
R^{uv}_{\a\b} &= \frac{1}{2h}\!\Eint \sum\nolimits_{\g\d}\nsum
	C^{u\a,\g\d}_{n,E}C^{v\b,\g\d\ast}_{n,E}
\end{align}
\end{subequations}
with
\begin{subequations}\label{TwMENDet}
\begin{align}
A^{uv,\a\b}_{n,E} &\equiv \xi^u_E\xi^v_E \d_{\a\b}\d_{n0}
	f'^\a_E-\xi^u_{\En}\xi^v_E \abs{S^{\a\b}_{\En,E}}^2f'^\b_E,\\
B^{uv,\a}_{n,E} &\equiv \gsum \xi^u_{\En}\xi^v_{\En} 
	\abs{S^{\a\g}_{\En,E}}^2(f'^\g_E -f'^\a_{\En}),\\
C^{u\a,\b\g}_{n,E} &\equiv\sum\nolimits_m \xi^u_{E\ix{m}}
		S^{\a\g}_{E\ix{m},\En}S^{\a\b\ast}_{E\ix{m},E}
		(f^\b_E-f^\g_{\En})
\end{align}
\end{subequations}
and $f'^\a_E \equiv f^\a_E(1-f^\a_E)$. 
The thermal, or Nyquist-Johnson, noise $D^{uv}_{\a\b}$ thereby arises 
from thermal fluctuations in the injected beams of carriers and 
vanishes in the zero-temperature limit.
By contrast, the shot noise $R^{uv}_{\a\b}$ stems from the 
probabilistic nature of carrier transmissions through the sample, and
therefore persists at zero temperature.

\subsection{Unitarity and Time-Reversal Symmetry}\label{SecSTUniTRS}
The Floquet scattering amplitudes generally depend on the structure of
the sample and the applied driving protocols $\df_t$, where $\df\equiv
\{V_j\}$ denotes the set of external control parameters.
Still, they obey two universal relations, which follow from 
fundamental principles. 
First, the unitarity conditions
\begin{subequations}\label{TwSumRulesSM} 
\begin{align}
\label{TwSumRulesSMl} 
& \asum\nsum S^{\a\b}_{E\ix{n},E}S^{\a\g\ast}_{\En,E\ix{m}}
	=\d_{\b\g}\d_{m0}\quad\text{and}\\[3pt]
\label{TwSumRulesSMr} 
& \asum\nsum S^{\b\a}_{E,\En}S^{\g\a\ast}_{E\ix{m},\En}
	=\d_{\b\g}\d_{m0}
\end{align}
\end{subequations}
ensure the conservation of probabilities in individual scattering 
events. 
Second, the invariance of Schr\"odinger's equation under time reversal
implies the symmetry
\begin{equation}\label{TwTimeReversal}
S^{\alpha\beta}_{E\ix{n},E} = \mathsf{T}_\mathbf{B}\mathsf{T}_{\df}
	S^{\beta\alpha}_{E,E\ix{n}},
\end{equation}
where the symbolic operators $\mathsf{T}_\mathbf{B}$ and 
$\mathsf{T}_{\df}$ indicate the reversal of external magnetic fields 
and driving protocols, respectively. 
Note that, while the unitarity conditions \eqref{TwSumRulesSM} apply 
to the scattering amplitudes of any given sample, the symmetry 
relation \eqref{TwTimeReversal} connects the scattering amplitudes of
two different systems that are related to each other by time reversal.

\subsection{Thermodynamics}
The thermodynamics of coherent transport can be developed from the
conservation laws
\begin{equation}\label{TwConsLaws}
\asum J^\rho_\a = 0 \quad\text{and}\quad
	\Pi_{{{\rm ac}}} + \asum J^\varepsilon_\a  = 0, 
\end{equation}
which can be easily verified using generalized Landauer-B\"uttiker
formula \eqref{TwMeanCurrents} and the unitarity conditions
\eqref{TwSumRulesSM}. 
They reflect the fact that neither matter nor energy are accumulated
in the sample over a full cycle. 
The quantity 
\begin{equation}\label{TwAcPower}
\Pi_{{{\rm ac}}}\equiv-\asum J^\varepsilon_\a = \frac{1}{\T}\Eint\absum\nsum
	n\abs{S^{\alpha\beta}_{E\ix{n},E}}^2 f^\beta_E
\end{equation}
thereby corresponds to the average mechanical power that is absorbed 
by the carriers from the driving fields. 

Combining the conservation laws \eqref{TwConsLaws}, leads to the 
first law of thermodynamics 
\begin{equation}\label{TwFirstLaw}
\Pi_{{{\rm ac}}} + \Pi_{{{\rm el}}} = -\asum J^q_\alpha. 
\end{equation} 
Here, the electrical power $\Pi_{{{\rm el}}}$, which is generated by
matter currents flowing in the direction of chemical potential 
gradients, and the heat currents entering the conductor from the 
reservoirs,  $J^q_\a$, are given by  
\begin{equation}\label{TwElPowerHeatCurr} 
\Pi_{{{\rm el}}}\equiv \asum \mu_\alpha J^\rho_\alpha
	\quad\text{and}\quad
	J^q_\a \equiv J^\ve_\a - \mu_\a J^\rho_\a. 
\end{equation} 
Since the transfer of carriers through the system is coherent, and 
thus reversible, dissipation occurs only in the reservoirs due to 
the in- and outflux of heat. 
Hence, the total rate of entropy production is given by  
\begin{align}\label{TwEntProd}
\sigma &\equiv -\asum J^q_\alpha/T_\alpha\\
	&= \frac{1}{h}\Eint \sum\nolimits_{\a\b}\nsum
	\left(\frac{\En-\mu_\a}{T_\a} - 
	\frac{E-\mu_\b}{T_\b}\right)
	\abs{S^{\a\b}_{\En,E}}^2 f^\b_E. 
	\nonumber
\end{align}
This expression is non-negative for any temperature and chemical 
potential profiles and any set of Floquet scattering amplitudes 
\cite{Brandner2020}. 
The scattering formalism is therefore consistent with the second
law of thermodynamics, which requires $\sigma\geq 0$.

\section{Linear Response}\label{SecFAR}
In irreversible thermodynamics, transport is described in terms of 
thermodynamic forces, or affinities, which correspond to gradients of
intensive variables, such as temperature, and currents of extensive 
quantities like energy \cite{Callen1985}. 
Every affinity forms a conjugate pair with a specific current such
that the products of these pairs add up to the total rate of entropy
production in the system.
Close to equilibrium, the currents become linear functions of the 
affinities with the corresponding response coefficients obeying two 
universal relations:
the Onsager-Casimir symmetry, which connects reciprocal coefficients
\cite{Onsager1931a,Onsager1931,Casimir1945}, and the 
fluctuation-dissipation theorem, which relates them to equilibrium 
current fluctuations \cite{Kubo1966,Kubo1998}.
These results are well-established for stationary coherent transport
in mesoscopic systems \cite{Sivan1986,Butcher1990,Buttiker1992}. 
In the following, we show how they can be extended to systems with 
periodic driving by further developing the approach that was
proposed in Ref.~\cite{Ludovico2015b}.

\subsection{Affinities}
The affinities for the matter and heat currents are given by the 
thermo-chemical gradients 
\begin{equation}\label{TwAffinities}
\F^\r_\a \equiv (\mu_\a-\mu)/T \quad\text{and}\quad
	\F^q_\a   \equiv 1/T - 1/T_\a,
\end{equation}
where $\mu$ and $T$ are the reference chemical potential and
temperature. 
Using these definitions and the conservation laws \eqref{TwConsLaws},
the total rate of entropy production \eqref{TwEntProd} can be 
written as 
\begin{equation}\label{TwEntProdAux}
\sigma = \Pi_{{{\rm ac}}}/T + \asum\xsum F^x_\a J^x_\a
\end{equation}
with $x=\rho,q$. 
Upon recalling the expression \eqref{TwAcPower} for the average 
mechanical power, this result suggests that we introduce the photon
current 
\begin{equation}\label{TwPhCurr}
J^\w \equiv \Pi_{{{\rm ac}}}/\hw = 
	\frac{1}{h}\Eint \sum\nolimits_{\a\b}\nsum n	
	\abs{S^{\a\b}_{\En,E}}^2 f^\b_E 
\end{equation}
and the corresponding affinity $F^\w\equiv \hw/T$ such that 
$\sigma$ assumes the canonical bilinear form 
\begin{equation}\label{TwSecondLawCan}
\sigma =  J^\w\F^\w
	 +\asum\xsum J^x_\a\F^x_\a = \Asum J_A F_A. 
\end{equation}
Here, we use capital roman letters to denote compound indices 
covering both thermo-chemical and mechanical quantities, i.e., 
$A=\{(x,\a)\},\w$.

Two remarks are in order. 
First, the interpretation of $J^\w$ as flux of photons derives from 
the fact that the carriers and the driving fields exchange only
discrete amounts of energy. 
This phenomenon, which, on the technical level, is a consequence of
the Floquet theorem \cite{Brandner2020}, is a manifestation of the
laws of quantum mechanics and has no counterpart in classical
mechanics \cite{Gritsev2017a}. 
Second, to achieve a thermodynamic unification of steady-state
and periodic driving, the driving frequency is treated as a
thermodynamic force in Eq.~\eqref{TwSecondLawCan}.
This approach, which was proposed in Ref.~\cite{Ludovico2015b},
is more suitable for coherent transport than using the amplitude of
the time-dependent fields as an effective affinity, a scheme that has
proved very useful for systems obeying stochastic dynamics
\cite{Brandner2015f,Proesmans2015a,Brandner2016,Proesmans2015,
Proesmans2016b}. 
In particular, as we show next, the frequency-based approach enables a
non-trivial linear-response theory, while a perturbation theory in the
driving strength only leads to a trivial decoupling of thermo-chemical 
currents and mechanical driving, see App.~\ref{Appx_KinCoeff_Wd} for 
details. 

\subsection{Kinetic Coefficients}\label{Sec_KinCoeff}
The kinetic coefficients that govern the relation between currents 
and affinities in the linear-response regime are defined as 
\begin{equation}\label{TwKinCoeffGen}
L_{AB} = \partial_{F_B}J_A\eq
\end{equation}
where the notation $\cdots\eq$ indicates the limit $F_A\rightarrow 0$.
To calculate the coefficients \eqref{TwKinCoeffGen}, we first observe
that the Fermi functions of the reservoirs are given by 
\begin{equation}\label{TwFermiFunctExp}
f^\a_E	\simeq f_E + \bigl(\F^\r_\a +(E-\mu)\F^q_\a\bigr)f'_E\\
\end{equation}
with $f_E\equiv f^\a_E\eq$ and $f'_E\equiv f'^\a_E\eq= f_E(1-f_E)$, up
to second-order corrections in the thermo-chemical affinities. 
Second, we recall that the Floquet scattering amplitudes admit the 
low-frequency expansion \cite{Moskalets2012,Thomas2012}
\begin{widetext}
\begin{equation}\label{TwAdApproxFSA}
S^{\a\b}_{E\ix{n},E\ix{m}}\simeq\frac{1}{\T} \tint \!
	\left(\AS^{\a\b}_{E,\df\ix{t}}	+\hw\frac{n+m}{2}
	\partial_E^{\phantom{\b}}\AS^{\a\b}_{E,\df\ix{t}}
	+\hw \A^{\a\b}_{E,t}\right)e^{i(n-m)\w t}. 
\end{equation}
\end{widetext}
Here, the frozen scattering amplitudes $\AS^{\a\b}_{E,\df}$ describe 
the transmission of carriers with energy $E$ at fixed parameters
$\df$. 
The correction term $\mathcal{A}^{\a\b}_{E,t}$ is required to ensure 
that the right-hand side of Eq.~\eqref{TwAdApproxFSA} obeys the 
unitarity condition \eqref{TwSumRulesSM}. 
In general, the approximation \eqref{TwAdApproxFSA} is 
applicable if the driving fields vary only slightly during the 
average dwell time $\tau_{{{\rm dw}}}$ of carriers inside the sample.
This time scale is connected to the typical energy range $\d_E$ over 
which the frozen scattering amplitudes change by the relation 
$\tau_{{{\rm dw}}}= \hbar/\d_E$, for details see
\cite{Texier2015,Moskalets2012,Gasparian1996}. 

Using Eqs.~\eqref{TwFermiFunctExp} and \eqref{TwAdApproxFSA}, the 
kinetic coefficients \eqref{TwKinCoeffGen} can be determined as 
\begin{subequations}\label{TwKinCoeffExp}
\begin{align}
\label{TwKinCoeffExpa}
L_{\a\b}^{xy} & = \frac{1}{h}\Eint \; \zeta^x_E\zeta^y_E
	\Bigl(\d_{\a\b}-
	\bigl\llangle\bigl|\AS^{\a\b}_{E,\df}\bigr|^2\bigr\rrangle\Bigr)
	f'_E ,\\
\label{TwKinCoeffExpb}
L_\a^{x\w} & = \frac{1}{h} \Eint\; \zeta^x_E\zeta^\w
	\bsum {{{\rm Im}}}\Bigl[\bigl\llangle
	\dot{\AS}_{E,\df}^{\a\b}\AS_{E,\df}^{\a\b\ast}\bigr\rrangle\Bigr]
	f'_E ,\\
\label{TwKinCoeffExpc}
L_\a^{\w x}  & = \frac{1}{h} \Eint\; \zeta^x_E\zeta^\w
	\bsum {{{\rm Im}}}\Bigl[\bigl\llangle
	\AS_{E,\df}^{\b\a}\dot{\AS}_{E,\df}^{\b\a\ast}\bigr\rrangle\Bigr]
	f'_E ,\\
\label{TwKinCoeffExpd}
L^{\w\w} & = \frac{1}{2h} \Eint\; (\zeta^\w)^2\absum \;
	\bigl\llangle\bigl|\dot{\AS}^{\a\b}_{E,\df}\bigr|^2\bigr\rrangle
	 f'_E
\end{align}
\end{subequations}
with $\zeta^\r_E\equiv 1,\;\zeta^q_E\equiv E-\mu,\;\zeta^\w\equiv
1/\w$, dots indicating time derivatives and double brackets denoting 
the time average over one period, $\llangle \cdots\rrangle \equiv 
\frac{1}{\T}\tint \cdots$, for details see 
App.~\ref{Appx_KinCoeff_Ad}.
Notably, the expressions \eqref{TwKinCoeffExp} do not depend on the 
corrections $\mathcal{A}^{\a\b}_{E,t}$ as a result of the unitarity 
condition \eqref{TwSumRulesSM}.
Instead, they involve only the frozen scattering amplitudes 
$\AS^{\a\b}_{E,\df}$, which are generally much easier to obtain than
the full Floquet scattering amplitudes. 
 
The linear response regime with respect to the affinities $F^\r_\a$, 
$F^q_\a$ and $F^\w$ is defined by the three conditions 
\begin{equation}
\F_\a^\r \ll \mu/T, \quad \F_\a^q \ll 1/T 
	\quad\text{and}\quad \F^\w\ll \d_E/T, 
\end{equation}
under which the currents obey the kinetic equations
\begin{equation}\label{TwARKinEq}
J_A = \Bsum L_{AB} F_B 
\end{equation}
This result extends the conventional framework of linear-irreversible
thermodynamics to periodically driven coherent conductors. 
We note that, at low temperatures, the function $f'_E$ in the
Eqs.~\eqref{TwKinCoeffExp} is sharply peaked around $\mu$. 
The transmission of carriers then occurs only at energies close to the
Fermi edge. 
It is therefore typically sufficient to require that the slow-driving 
condition $F^\w\ll\d_E/\hbar$ is obeyed at $E\simeq\mu$ for the kinetic
equations \eqref{TwARKinEq} to be valid. 

\subsection{Onsager-Casimir Relations}\label{Sec_OCRel}
The Onsager-Casimir, or reciprocal, relations between linear-response
coefficients follow from the symmetry of microscopic dynamics under 
time reversal. 
This principle enters stationary scattering theory through the 
property 
\begin{equation}\label{TwARTimeReversal}
\AS^{\a\b}_{E,\df} = \mathsf{T}_\mathbf{B}\AS^{\b\a}_{E,\df}
\end{equation}
of the frozen scattering amplitudes \cite{Brandner2020}, which, 
together with the formulas \eqref{TwKinCoeffExp}, implies the
relations
\begin{equation}\label{TwARReciprocity}
\TB L^{xy}_{\a\b} = L^{yx}_{\b\a}\quad\text{and}\quad
	\TB L^{\w x}_\a = - L^{x\w}_\a;
\end{equation}
recall Sec.~\eqref{SecSTUniTRS} for the definition of 
$\mathsf{T}_\mathbf{B}$ and $\mathsf{T}_{\df}$.
Hence, while the thermo-chemical coefficients obey the conventional
Onsager-Casimir symmetry, the cross-coefficients that couple either
to the mechanical affinity or the photon current are anti-symmetric.
The original symmetry can, however, be restored for all kinetic 
coefficients by reversing both magnetic fields and driving protocols. 
That is, we have the generalized Onsager-Casimir relations 
\footnote{
Note that
\begin{align*}
& \TV \bigl\llangle \abs{\AS^{\a\b}_{E,\df}}^2\bigr\rrangle
	=\bigl\llangle \abs{\AS^{\a\b}_{E,\df}}^2\bigr\rrangle 
	\quad\text{and}\\
& \TV \bigl\llangle \dot{\AS}^{\a\b}_{E,\df}
	\AS^{\a\b\ast}_{E,\df}\bigr\rrangle
	=-\bigl\llangle \dot{\AS}^{\a\b}_{E,\df}
	\AS^{\a\b\ast}_{E,\df}\bigr\rrangle,
\end{align*}
since $\TV \AS^{\a\b}_{E,\df\ix{t}} = \AS^{\a\b}_{E,\df_{-t}}$. 
}
\begin{equation}\label{TwARReciprocity2}
\TB\TV L_{AB} = L_{BA}
\end{equation}
Notably, this result implies that, if the driving protocols are 
symmetric, i.e., if $\df_t = \df_{-t}$, the mechanical and 
thermo-chemical currents and affinities decouple, since 
$L^{x\w}_\a= L^{\w x}_\a =0$. 

\subsection{Fluctuation-Dissipation Theorem}
The fluctuation-dissipation theorem provides a link between kinetic 
coefficients and equilibrium current fluctuations. 
For periodically driven coherent conductors, this connection can be 
established as follows.
First, we note that, with respect to the fluctuations of matter and 
energy currents, $P^{uv}_{\a\b}$, which are spelled out in 
Eqs.~\eqref{TwMEN} and \eqref{TwMENDet}, the joint fluctuations
of matter and heat currents are given by 
\begin{align}\label{LRFDT_CurrFluct} 
P^{xy}_{\a\b} = \sum\nolimits_{uv} c^{xu}_\a c^{yv}_\b 
	P^{uv}_{\a\b}
\end{align}
with $c^{\r u}_\a\equiv\d_{u\r}$ and $c^{q u}_\a\equiv\d_{u\ve}-\mu_\a
\d_{u\r}$ \cite{Brandner2020}. 
Next, we observe that the shot noise \eqref{TwMENSh} vanishes in 
equilibrium, $R^{uv}_{\a\b}|_{{{\rm eq}}}=0$. 
Therefore, we have 
\begin{equation}
P^{xy}_{\a\b}|_{{{\rm eq}}}= \sum\nolimits_{uv}(c^{xu}_\a c^{yv}_\b
	D^{uv}_{\a\b})|_{{{\rm eq}}} = L^{xy}_{\a\b} + L^{yx}_{\b\a}
\end{equation}
as can be easily verified by inspection. 
For systems without magnetic fields, we thus recover the standard 
result $P^{xy}_{\a\b}=2L^{xy}_{\a\b}$ by using the symmetry
\eqref{TwARReciprocity}. 

To derive a fluctuation-dissipation relation for the coefficients 
$L^{\w x}_\a$, $L^{x\w}_\a$ and $L^{\w\w}$, we have to consider the
fluctuations involving the photon current, which, due to energy
conservation, can be obtained from the sum rules
\begin{align}
P^{u\w}_\a &= -\bsum P^{u\ve}_{\a\b}/\hw
	\;\;\text{and}\;\;
P^{\w\w}   = \absum P^{\ve\ve}_{\a\b}/(\hw)^2.
\end{align}
Inserting the expressions \eqref{TwMEN} and \eqref{TwMENDet} for the 
thermal and the shot noise and using the unitarity conditions
\eqref{TwSumRulesSM} for the Floquet scattering
amplitudes yields the explicit results $P^{u\w}_\a = D^{u\w}_\a +
R^{u\w}_\a$ and $P^{\w\w}= D^{\w\w} + R^{\w\w}$ with 
\begin{subequations}
\begin{align}
D_\a^{u\w} & = \frac{1}{h}\Eint \bsum\nsum 
	\bigl( \xi^u_E A^{\w1,\b\a}_{n,E}
	-\xi^u_{\En} A^{\w1,\a\b}_{n,E}\bigr),\\
D^{\w\w} &=\frac{1}{h}\Eint \absum\nsum A^{\w2,\a\b}_{n,E},\\
R^{u\w}_\a &= \frac{1}{2h}\Eint \sum\nolimits_{\g\d}\nsum 
	C^{u\a,\g\d}_{n,E} C^{\w,\g\d\ast}_{n,E},\\
R^{\w\w} &= \frac{1}{2h}\Eint \sum\nolimits_{\g\d}\nsum
	C^{\w,\g\d}_{n,E}C^{\w,\g\d\ast}_{n,E}
\end{align}
\end{subequations}
and 
\begin{subequations}
\begin{align}
A^{\w k,\a\b}_{n,E} & \equiv n^k \abs{S^{\a\b}_{\En,E}}^2 f'^\b_E\\
\label{TwNMechAuxC}
C^{\w,\b\g}_{n,E} &\equiv\asum\sum\nolimits_m m S^{\a\g}_{E\ix{m},\En}
	S^{\a\b\ast}_{E\ix{m},E}(f^\b_E- f^\g_{\En}). 
\end{align}
\end{subequations}
The expression \eqref{TwNMechAuxC} shows that $R^{u\w}_\a\eq=0$. 
After switching from energy to heat currents, we are therefore left 
with 
\footnote{
Recall that 
\begin{align*}
&\nsum n \abs{S^{\a\b}_{\En,E}}^2 \rightarrow \zeta^\w {{{\rm Im}}}
	\Bigl[\bigl\llangle \AS^{\a\b}_{E,\df} \dot{\AS}^{\a\b\ast}_{E,\df}
	\bigr\rrangle\Bigr] \quad\text{and}\\
&\nsum n^2 \abs{S^{\a\b}_{\En,E}}^2 \rightarrow (\zeta^\w)^2 
	\bigl\llangle \dot{\AS}^{\a\b}_{E,\df} 
	\dot{\AS}^{\a\b\ast}_{E,\df}\bigr\rrangle
\end{align*}
in the limit $\w\rightarrow 0$.}
\begin{subequations}
\begin{align}
\label{TwFDTThChemMech}
P^{x\w}_\a|_{{{\rm eq}}} 
	& = \sum\nolimits_u (c^{xu}_\a D^{u\w}_\a)\eq= L^{x\w}_\a+L^{\w x}_\a
	\quad\text{and}\\
P^{\w\w}\eq & = D^{\w\w}\eq = 2 L^{\w\w}. 
\end{align}
\end{subequations}
Hence, quite remarkably, the mechanical kinetic coefficients and the 
equilibrium fluctuations of the photon current obey the same relations
as their thermo-chemical counterparts. 
This result is summarized by the extended fluctuation-dissipation 
theorem
\begin{equation}\label{TwFDT} 
P_{AB}\eq = L_{AB} + L_{BA}, 
\end{equation}
which completes our linear-response framework. 
Notably, it implies, together with the symmetries
\eqref{TwARReciprocity} and \eqref{TwARReciprocity2}, that 
$P^{x\w}_\a|_{{{\rm eq}}}=0$ for systems without a magnetic field or
with symmetric driving protocols.
That is, equilibrium correlations between the photon current and
the thermo-chemical currents are ultimately a result of broken 
time-reversal symmetry. 

\section{New Bounds}\label{SecTC}
According to the second law, the rate of entropy production $\sigma$,
which provides a measure for the thermodynamic cost of irreversible 
transport, cannot be negative. 
However, the laws of thermodynamics do not determine how much entropy
must be generated to sustain a given current as the following argument
shows. 
In linear response, the currents $J_A$ can be divided into an 
irreversible and a reversible contributions given by 
\cite{Brandner2013}
\begin{equation}
J_A^{{{\rm irr}}}\equiv \Bsum \frac{L_{AB}+L_{BA}}{2}F_B, 
	\quad
	J_A^{{{\rm rev}}}\equiv \Bsum \frac{L_{AB}-L_{BA}}{2}F_B.
\end{equation}
Using these variables, the rate of entropy production 
\eqref{TwSecondLawCan}, can be expressed as 
\begin{equation}\label{LRNB_EntProd}
\sigma = \Asum J_A^{{{\rm irr}}}F_A. 
\end{equation}
Hence, the reversible currents, which, due to the generalized 
Onsager-Casimir relation \eqref{TwARReciprocity2}, exist only in 
systems with broken time reversal symmetry, do not contribute to 
$\sigma$. 
As a result, transport without dissipation seems to be possible
in situations where $J_A^{{{\rm irr}}}=0$ and at least one reversible
current is finite \cite{Benenti2011,Dechant2018b}. 
This \emph{a priori} surprising observation prompts the question 
whether their might be stronger bounds on thermal currents than the
second law. 
In the following, we first derive such bounds for coherent transport
and then prove their robustness against dephasing.

\subsection{Coherent Transport}\label{SecThNB}

Our new bounds follow from the unitarity conditions for the frozen 
scattering amplitudes \cite{Moskalets2012}, 
\begin{equation}\label{TwARSumRulesSM}
\asum \AS^{\a\b}_{E,\df}\AS^{\a\g\ast}_{E,\df} 
	=\asum \AS^{\b\a}_{E,\df}\AS^{\g\a\ast}_{E,\df} = \d_{\b\g},
\end{equation}
which ensure probability conservation, and the sum rule
\begin{equation}\label{TwARSumRulesSM2}
\absum \bigl\llangle \dot{\AS}^{\a\b}_{E,\df}\AS^{\a\b\ast}_{E,\df}
	\bigr\rrangle =0,
\end{equation}
which plays the role of a gauge condition fixing the global phase of 
the frozen scattering amplitudes, see Lemma~\ref{Lem_3} of 
App.~\ref{App_Lemmas}. 
Together, they lead to the sum rules $\asum L^{xy}_{\a\b}=0$ and
$\asum L^{x\w}_\a=0$ for the kinetic coefficients 
\eqref{TwKinCoeffExp}, which imply the conservation laws 
\begin{equation}\label{LRNB_ConsLaws}
\asum J^\rho_\a = 0 \quad\text{and}\quad\asum J^q_\a=0. 
\end{equation}
Note that the conservation law for the heat currents is consistent 
with the first law \eqref{TwFirstLaw} in linear response, since the
electrical and the mechanical power, $\Pi_{{{\rm el}}}$ and 
$\Pi_{{{\rm ac}}}$, are of second order in the affinities. 

As a technical tool, we now define the quadratic form 
\begin{equation}\label{ThrNBQuadForm}
\Xi\equiv \sigma +\asum\xsum J^x_\a G^x_\a 
	+\asum\xysum K^{xy}G^x_\a G^y_\a,
\end{equation}
where the $G^x_\a$ are real but otherwise arbitrary variables and the
coefficients 
\begin{equation}\label{ThNBKCoeff}
K^{xy}\equiv \frac{1}{2h}\Eint\; \zeta^x_E\zeta^y_E f'_E
\end{equation}
have been chosen such that $\Xi$ is positive semi-definite. 
To verify this property, we expand the rate of entropy production 
$\sigma$ and the thermo-chemical currents $J^x_\a$ in the affinities
using Eq.~\eqref{TwSecondLawCan} and the kinetic equations 
\eqref{TwARKinEq}.
Upon inserting the expressions \eqref{TwKinCoeffExp} for the kinetic
coefficients, and applying the unitarity conditions 
\eqref{TwARSumRulesSM}, Eq.~\eqref{ThrNBQuadForm} can thus be
rewritten in the form
\begin{equation}
\Xi = \frac{1}{2h}\Eint\;\absum \bigl\llangle\abs{
	X^{\a\b}_E\AS^{\a\b}_{E,\df}- iX^\w\dot{\AS}^{\a\b}_{E,\df}}^2
	\bigr\rrangle f'_E 
\end{equation}
with $X^{\a\b}_E\equiv \xsum \zeta^x_E (F^x_\a - F^x_\b +G^x_\a)$ and 
$X^\w\equiv\zeta^w F^\w$, which proves that $\Xi$ cannot be negative, 
since $f'_E\geq 0$. 

This result leads to a whole family of bounds on $\sigma$, which can
be extracted as follows. 
We first set $G^q_\a=0$ and minimize the right-hand side of 
Eq.~\eqref{ThrNBQuadForm} with respect to the variables $G^\r_\a$. 
After repeating this step with the roles of $G^q_\a$ and $G^\r_\a$ 
interchanged, we arrive at the cumulative bound
\begin{equation}\label{ThNBCompoundBnd}
\sigma \geq \frac{1}{4 K^{xx}}\asum (J^x_\a)^2.
\end{equation} 
To obtain bounds that involve only a single current, we rewrite
Eq.~\eqref{ThNBCompoundBnd} as 
\begin{equation}
\sigma\geq \frac{(J^x_\a)^2}{4K^{xx}}
	+ M^x_\a \quad\text{with}\quad
	M^x_\a \equiv \frac{1}{4 K^{xx}}
	\sum\nolimits_{\b\neq\a} (J^x_\b)^2.
\end{equation}
Minimizing $M^x_\a$ with respect to the currents $J^x_\b$ while
taking into account the conservation laws \eqref{LRNB_ConsLaws} as a 
constraint yields $M^x_\a \geq (J^x_\a)^2/4K^{xx}(N-1)$ and thus 
\begin{equation}\label{ThNB}
\sigma \geq \frac{N}{4K^{xx}(N-1)} (J^x_\a)^2,
\end{equation}
where $N$ is the number of terminals of the conductor and the 
inequality holds for any pair of indices $x$ and $\a$. 

This relation constitutes a key result of this paper. 
Going beyond the second law, it shows that, regardless of the behavior
of the system under time reversal, any matter or heat current comes at
the price of a minimal rate of entropy production that is proportional
to the square of this current. 
The positive coefficients $K^{xx}$ thereby depend only on the equilibrium 
properties of the reservoirs. 
Specifically, they are given by 
\begin{subequations}\label{ThNBPropPara}
\begin{align}
\label{ThNBPropParaa}
K^{\r\r}&= \frac{T\varphi}{2h(1+\varphi)}\leq \frac{T}{2h}
	\quad\text{and}\\
K^{qq} &= 
	\frac{\pi^2 T^3}{6h}
	\begin{aligned}[t]
	&-\frac{T^3(\ln[\varphi])^2}{2h(1+\varphi)}
	-\frac{T^3\ln[\varphi]\ln[1+1/\varphi]}{h}\\
	&+\frac{T^3{{{\rm Li}}}_2[-1/\varphi]}{h} 
	\leq \frac{\pi^2 T^3}{6h},
	\end{aligned}
\end{align}
\end{subequations}
where ${{{\rm Li}}}_2$ denotes the dilogarithm, $\varphi\equiv
\exp[\mu/T]$ the equilibrium fugacity of the reservoirs and the 
inequalities are saturated in the limit $\varphi\rightarrow\infty$, 
which is practically realized in mesoscopic conductors. 

\subsection{Dephasing}
Coherent transport is characterized by a fixed phase relation between
incoming and outgoing carriers. 
Under realistic conditions, however, phase-breaking mechanisms such as
carrier-carrier or carrier-phonon interactions can hardly be
completely suppressed. 
Probe terminals provide an elegant way to account for such effects
\cite{Buttiker1988}. 
In this approach, virtual reservoirs are attached to the conductor, 
whose temperature and chemical potential are adjusted such that they
do not exchange matter or heat with the remaining system on average, 
but rather act as a source of dephasing. 

On the technical level, the virtual reservoirs differ from physical 
ones only in that their affinities are fixed by the conditions of zero
mean currents. 
The bounds \eqref{ThNBCompoundBnd}, however, were derived without any
assumptions on the affinities. 
They therefore apply also to systems with arbitrary many probe 
terminals, where only the currents between physical reservoirs 
contribute to the sum on the right; 
those into the virtual reservoirs are zero by construction. 
Since the conservation laws \eqref{LRNB_ConsLaws} are likewise not 
affected by the probe terminals, also the bounds \eqref{ThNB} remain 
valid with $N$ referring to the number of real terminals. 
Hence, our new bounds are robust against dephasing and hold even in
the limit of fully incoherent transmission \cite{Buttiker1988}.

\section{Applications}\label{SecAdChrPump}
In this section we discuss three different quantum devices to explore
the practical implications of our new bounds on coherent transport and
their potential as tools of thermodynamic inference. 
As a first example, we consider a basic model of a magnetic-flux 
driven quantum generator, which was proposed in Ref.~\cite{Avron2001}.
This case study serves as a simple illustration of our general theory
and shows that our bounds are tight. 
We then move on to parametric quantum pumps, which make it possible to
move a well-defined amount of carriers from one reservoir to another 
in a given cycle time.
Such devices can be realized for instance with tunable-barrier quantum
dots \cite{Rossi2014,Yamahata2016}, and, owing to their high accuracy,
are promising candidates for experimentally accessible quantum 
representations of the ampere \cite{Giblin2012,Pashkin2013,Jehl2013}.
Here, we show that, in the slow-driving regime, the energy that is
required to move a given amount of carriers is subject to a 
fundamental lower bound, which depends only on the cycle time. 
We further derive an explicit optimization principle for adiabatic
quantum pumps by connecting our theory with the geometric approach to 
parametric pumping \cite{Brouwer1998,Aleiner1998,Shutenko2000,
Avron2000,Levinson2001,Sinitsyn2007,Potanina2019} and the notion of
thermodynamic length
\cite{Weinhold1975,Salamon1983,Salamon1985,Andresen1988,Crooks2007}. 
As a third application of our theory, we derive a universal trade-off
relation between the efficiency and the power consumption of adiabatic
quantum motors, that is, devices that convert an electric current into
motive power of mesoscopic mechanical objects like nano-paddle wheels
or conveyor belts \cite{Bode2011,Bustos-Marun2013,Arrachea2016,
Bruch2018a}.

\subsection{Quantum Generator}\label{SecAdChrPumpTM}

\begin{figure}
\includegraphics[scale=1.05]{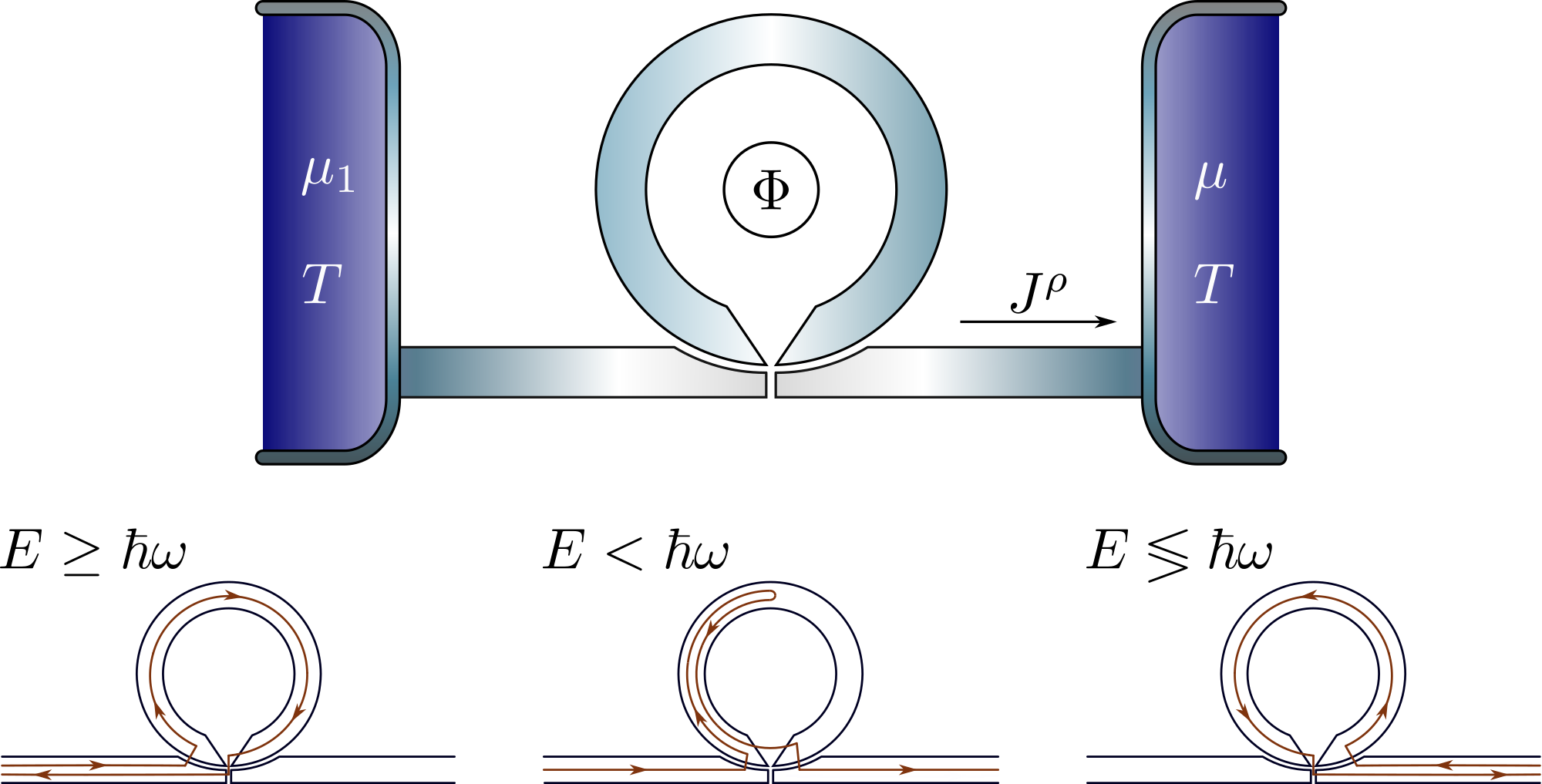}
\caption{Quantum generator. 
A mesoscopic ring is connected to two leads via a four-way beam 
splitter such that, without a magnetic flux, all incoming carriers are
reflected after passing through the ring.
The current $J^\r$ is generated by changing the magnetic flux $\Phi$
such that clockwise-moving carriers are decelerated by the emerging 
electromagnetic force around the ring. 
As a result, incoming carriers from the left are reflected only if 
their energy $E$ is above the threshold $\hw$, which is required to 
pass through the entire loop; 
carriers with $E<\hw$ are flipped around and transmitted to the right.
At the same time, carriers coming from the right are accelerated and
thus reflected regardless of their energy.
\label{Fig_QuantGen}}
\end{figure}

The setup of Fig.~\ref{Fig_QuantGen} provides a simple realization of
a quantum generator. 
The frozen scattering amplitudes for this system are 
\begin{align}\label{FrABScatAmpl}
\AS^{11}_{E,\phi} &= e^{i(\chi_E+\phi)}, \quad
\AS^{12}_{E,\phi} = \AS^{21}_{E,\phi} =0 \quad\text{and}\\
\AS^{22}_{E,\phi} &= e^{i(\chi_E-\phi)}, 
\nonumber
\end{align}
where the irrelevant dynamical phase $\chi_E$ is determined by the
circumference of the loop and the Aharonov-Bohm phase $\phi\equiv e
\Phi/\hbar c$ plays the role of an external control parameter;
here, $\Phi$ is the tunable magnetic flux through the ring, $e$ is the
carrier charge and $c$ the speed of light \cite{Nazarov2009b}. 
We use the right reservoir as a reference. 
Hence, the chemical affinity is $F^\r_1\equiv F^\r$ and the electric
current $J^\r_1\equiv J^\r$ flows from left to right. 
For a linearly increasing flux, i.e., for $\phi_t\equiv\w t$, the 
kinetic coefficients are
\begin{align}\label{AQG_KinCoeff}
L^{\r\r} & = 0, \quad
L^{\r\w} =- L^{\w\r} = \frac{T\varphi}{h(1+\varphi)}	\quad\text{and}\\
L^{\w\w} & = \frac{T\varphi}{h(1+\varphi)},
\nonumber
\end{align}
according to the formulas \eqref{TwKinCoeffExp}, where $\varphi=
\exp[\mu/T]$ is the fugacity of the right reservoir. 
The electric current and the rate of entropy production are thus given
by $J^\r= L^{\r\w}F^\w$ and $\sigma= L^{\w\w}(F^\w)^2$. 
Upon recalling the expressions \eqref{ThNBPropPara} for the 
coefficient $K^{\r\r}$, it is now straightforward to verify that the 
bound \eqref{ThNB} is saturated, that is, 
$\sigma = (J^\r)^2/2 K^{\r\r}$ for any $F^\r$ and $F^\w$. 
This result shows that our bounds are tight. 

On the microscopic level, the saturation of the bound \eqref{ThNB} is
a consequence of the working mechanism of the quantum generator, which
is described in Fig.~\ref{Fig_QuantGen}. 
Every transmitted carrier leads to the net dissipation of one quantum
of energy $\hw$. 
As a result, the irreversible part of the photon current, 
$J^\w_{{{\rm irr}}}= L^{\w\w} F^\w$, is equal to the electric current
$J^\r$ and proportional to the rate of entropy production $\sigma=F^\w
J^\w_{{{\rm irr}}}$. 

By contrast, the reversible photon current, $J^\w_{{{\rm rev}}}=
L^{\w\r}F^\r$, is decoupled from the dissipation rate $\sigma$, which
is independent of the chemical affinity $F^\r$.
In particular, for $F^\w=0$ and $F^\r\neq 0$, we have $\sigma=0$
and $J^\w=J^\w_{{{\rm rev}}}\neq 0$.  
This observation does not imply the occurrence of dissipationless 
transport, since the electric current vanishes for $F^\w=0$. 
It shows, however, that no general bound of the form \eqref{ThNB} 
exists for the photon current. 

\subsection{Parametric Quantum Pumps}\label{SecAdChrPumpTBP}

\subsubsection{Performance Bound}
A parametric quantum pump can be described as a two-terminal conductor,
whose potential landscape is changed periodically to generate a flow
of carriers between two reservoirs with the same chemical potential 
and temperature. 
One pumping cycle requires the energy input
$U\equiv\T\Pi_{{{\rm ac}}}$ and moves the amount of carriers 
$\mathcal{Q}\equiv\T J^\r$ from the first reservoir to the second, see
Fig.~\ref{Fig_QPump}. 
Our bound \eqref{ThNB} implies that, in the slow-driving regime, these
two figures are connected by the trade-off relation 
\begin{equation}\label{ApplAQP_NB}
U\geq \frac{T \mathcal{Q}^2}{2\T K^{\r\r}} \geq 
	\hbar\w \mathcal{Q}^2,
\end{equation}
where we have used that $U=\T T\sigma$ and the second inequality 
follows from Eq.~\eqref{ThNBPropParaa}.
This result is quite remarkable as it puts a universal lower bound on 
the energy that must be provided to generate a given pump flux 
$\mathcal{Q}$ in a given cycle time $\T$. 
Hence, Eq.~\eqref{ApplAQP_NB} makes it possible to estimate the
energy consumption of a quantum pump, even in situations, where only
the flux $\mathcal{Q}$ can be measured and the scattering amplitudes
of the sample are unknown. 

\subsubsection{Geometry and Optimal Driving Speed}
Finding optimal driving protocols for a quantum pump is a difficult 
task, which typically requires the solution of involved variational 
problems. 
In the slow-driving regime, it is, however, possible to derive a 
universal optimization principle for the parameterization $\g$ of the 
closed path $\Gamma$ that is mapped out by the vector of control
parameters $\df_t$ during the cycle. 
To this end, we recall that, by using Eq.~\eqref{TwKinCoeffExpc},
the pump flux can be written as a line integral in the space of
control parameters \cite{Brouwer1998},
\begin{equation}\label{ApplAQP_GeoCharge}
\mathcal{Q} = \T L^{\r w}F^\w =\oint_{\Gamma}
	\sum\nolimits_{j} \mathcal{A}^j_\df dV_j,
\end{equation}
which proves that $\mathcal{Q}$ is independent of the parametrization 
of $\Gamma$. 
The objects 
\begin{equation}\label{ApplAQQ_VectPot}
\mathcal{A}^j_\df\equiv \frac{1}{2\pi T} \Eint \bsum {{{\rm Im}}}
	\Bigl[S^{1\b\ast}_{E,\df} \partial_{V_j} S^{1\b}_{E,\df}\Bigr] f'_E
\end{equation}
are thereby considered as the components of a vector field
$\boldsymbol{\mathcal{A}}_{\df}$, which, in analogy to the theory of 
geometric phases, is called the Berry potential \cite{Berry1984}. 
Second, the input $U$ can be expressed in the form
\begin{equation}\label{ApplAQP_KinEnt}
U  = T \T L^{\w\w}(F^\w)^2 
	= h \tint \!\sum\nolimits_{ij} g^{ij}_{\df_t} 
	\dot{V}_{i,t} \dot{V}_{j,t},
\end{equation}
where we have used Eq.~\eqref{TwKinCoeffExpd} and introduced the 
thermodynamic metric \footnote{Note that the coefficients 
$g^{ij}_{\df}$ form a symmetric, positive semi-definite matrix
and can therefore be consistently identified with a, possibly 
degenerate, metric in the space of control parameters.}
\begin{equation}\label{ApplAQP_Metric} 
g^{ij}_\df\equiv \frac{1}{8\pi^2 T}\Eint \absum {{{\rm Re}}}
	\Bigl[\bigl(\partial_{V_i} S^{\a\b}_{E,\df}\bigr) 
	\bigl(\partial_{V_j} S^{\a\b\ast}_{E,\df}\bigr)
	\Bigr]f'_E. 
\end{equation}

The optimal parameterization $\g^\ast_t$ of the path $\Gamma$, which 
minimizes the input $U$, can be determined from the objective 
functional
\begin{equation}\label{ApplAQP_ObjF}
O[\dot{\hat{\g}}_t]\equiv \tint \! \left(\sum\nolimits_{ij} 
	g^{ij}_{\df_t}\dot{V}_{i,t} \dot{V}_{j,t}\bigl/\dot{\hat{\g}}_t
	-\Lambda\dot{\hat{\g}}_t\right),
\end{equation}
where $\Lambda$ is a Lagrange multiplier accounting for the boundary 
condition $\g_{\T} - \g_0 = \T$ and $\dot{\hat{\g}}_t$ is the 
derivative of the inverse function $\hat{\g}_t$ of $\g_t$ 
\footnote{Note that $\g_t$ is a monotonically increasing and therefore
invertible function.}.
Solving the Euler-Lagrange equation for this functional yields the 
condition
\begin{equation}\label{ApplAQP_OptP} 
t = \frac{\T}{\mathcal{L}} \int_0^{\g^\ast_t} \!\!\! ds \; 
	\sqrt{\sum\nolimits_{ij} g^{ij}_{\df_s} 
	\dot{V}_{i,s} \dot{V}_{j,s}},
\end{equation}
where
\begin{equation}
\mathcal{L}\equiv \oint_\Gamma \sqrt{\sum\nolimits_{ij}
	g^{ij}_{\df} dV_i dV_j}
\end{equation}
denotes the thermodynamic length of $\Gamma$. 
Replacing the protocols ${\df}_t$ with ${\df}_{\g^\ast_t}$ minimizes
the energy consumption of the pump without changing its flux 
$\mathcal{Q}$, which depends only on the path $\Gamma$. 
Inserting the minimal energy uptake $U^\ast=\hbar\w\mathcal{L}^2$ into
Eq.~\eqref{ApplAQP_NB} yields the remarkably simple relation 
\begin{equation}
\mathcal{L}\geq |\mathcal{Q}|
\end{equation}
between thermodynamic length and pump flux, which connects our 
thermodynamic bounds with the geometric theory of slowly driven 
quantum pumps. 

\subsubsection{Tunable-Barrier Pump}\label{Sec_TunBarPump}
\begin{figure}
\includegraphics[scale=1.05]{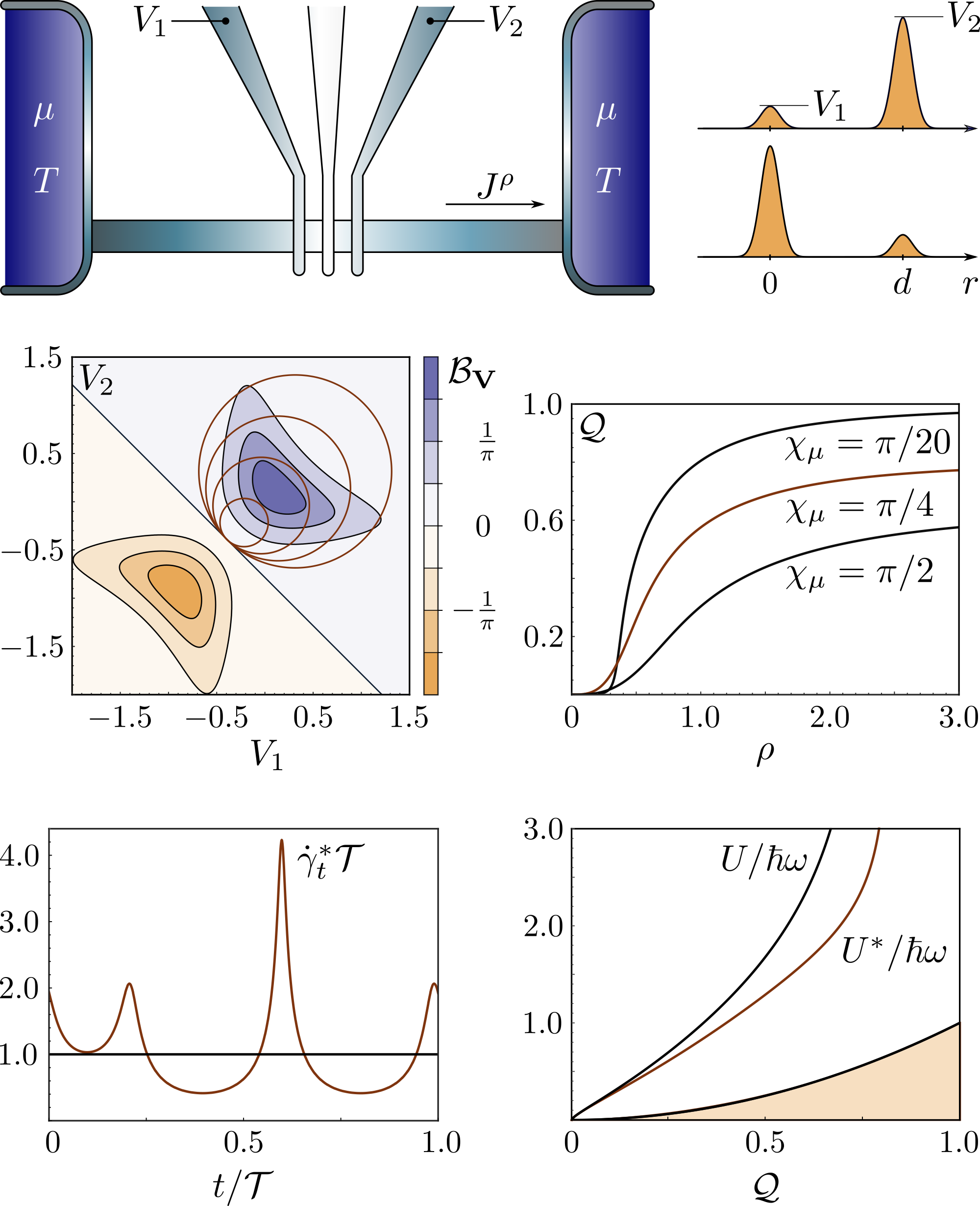}
\caption{Parametric quantum pump. 
\textbf{Top:}
Sketch of a generic setup. 
A narrow conductor connects two reservoirs with the same chemical 
potential and temperature. 
Periodically changing the gate voltages $V_1$ and $V_2$ creates two
oscillating potential barriers driving the pump current $J^\r$. 
\textbf{Middle:}
The left plot shows the Berry curvature $\mathcal{B}_\df$ for
$\chi_\mu =\pi/4$.
Circles indicate the control path that is determined by the 
protocols \eqref{FrTBDrPt} for $\r=1/4, \; 1/2, \; 3/4, \;1$. 
As the path expands into the positive peak of $\mathcal{B}_{\df}$,
the pump flux $\mathcal{Q}$ first increases sharply and then
approaches a limit value, which depend on the parameter $\chi_\mu$.
\textbf{Bottom:}
The left panel shows the optimal driving speed $\dot{\g}^\ast_t$ for
$\chi_\mu =\pi/4$ and $\r=1$. 
The horizontal line corresponds to constant speed. 
On the right, the energy uptake is plotted against the pump flux for
constant and optimal driving speed, $U$ and $U^\ast$ respectively,
where $\chi_\mu=\pi/4$ and $\r$ varies between $0$ and $5$.
The shaded area indicates the bound \eqref{ApplAQP_NB}.
\label{Fig_QPump}}
\end{figure}

We now consider a simple model of a quantum pump, where the potential
inside the conductor consists of two delta-barriers with
dimensionless strengths $V_1$ and $V_2$.
The single-particle Hamiltonian of this system reads
\begin{equation}
H_{\df}= \frac{p^2}{2M} + \frac{\hbar^2V_1}{Md}\d_{r} 
	+ \frac{\hbar^2 V_2}{Md}\d_{r-d},
\end{equation}
where $\df\equiv(V_1,V_2)$,  $p$ and $r$ are the momentum and position
of the carrier, $d$ denotes the distance between the two barriers and 
$M$ the carrier mass, see Fig.~\ref{Fig_QPump}.
The corresponding frozen scattering amplitudes are 
\cite{Moskalets2004}
\begin{subequations}\label{FrTBScatAmpl}
\begin{align}
&\AS^{12}_{E,\df} = \AS^{21}_{E,\df} 
	=\mathcal{Z}_{E,\df} \chi^2_E e^{i\chi_E},\\
&\AS^{11}_{E,\df} =\mathcal{Z}_{E,\df}\bigl(
	\l_1(\l_2-i\chi_E)-\l_2(\l_1+i\chi_E)e^{2i\chi_E}\bigr),\\
&\AS^{22}_{E,\df} =\mathcal{Z}_{E,\df}\bigl(
	\l_2(\l_1-i\chi_E)-\l_1(\l_2+i\chi_E)e^{2i\chi_E}\bigr)
\end{align}
\end{subequations}
with $\mathcal{Z}_{E,\df}\equiv 1\bigl/\bigl(\l_1\l_2e^{2i\chi_E}
-(\l_1-i\chi_E)(\l_2-i\chi_E)\bigr)$
and $\chi_E\equiv d\sqrt{2ME}/\hbar$.  
In the following, we focus on the low-temperature limit, where the 
function $f'_E$ is sharply peaked around $E=\mu$ and can therefore be 
replaced with $T\delta_{E-\mu}$ in the Eqs.~\eqref{ApplAQP_GeoCharge}
and \eqref{ApplAQP_Metric}. 

To find a suitable control path, we recall that the expression
\eqref{ApplAQP_GeoCharge} for the pump flux can be rewritten as an 
area integral with the help of Stokes' theorem \cite{Brouwer1998},
\begin{equation}\label{ApplAQP_BCurCharge} 
\mathcal{Q} = \int_{S_\Gamma} \!\!\! dS \; \mathcal{B}_{\df}, 
\end{equation}
where $\mathcal{B}_\df\equiv \partial_{V_1}\mathcal{A}^2_\df
-\partial_{V_2}\mathcal{A}^1_{\df}$ is the Berry curvature 
corresponding to the potential $\boldsymbol{\mathcal{A}}_\df$ and
$S_\Gamma$ denotes the area encircled by $\Gamma$. 
As shown in Fig.~\ref{Fig_QPump}, the function $\mathcal{B}_{\df}$
features two antisymmetric peaks. 
Hence, to generate a large flux $\mathcal{Q}$, the path $\Gamma$ has
to cover the positive peak while avoiding the negative one. 
This condition is met by the circles with the parameterization 
\begin{equation}\label{FrTBDrPt}
\l_{1,t} = V_0 - \r\cos[\w t], \;\;\;\;
\l_{2,t} = V_0 - \r\sin[\w t],
\end{equation}
where $V_0\equiv \r/\sqrt{2}-\chi_\mu\cot[\chi_\mu]/2$ and the radius
$\rho$ determines the amplitude of the driving.

Using the protocols \eqref{FrTBDrPt}, we numerically calculate the
flux $\mathcal{Q}$, the input $U$, the optimal driving speed 
$\dot{\g}^\ast$ and the minimized input $U^\ast$.
The results of these calculations are plotted in Fig.~\ref{Fig_QPump}. 
Two observations stand out. 
First, the energy consumption of the pump can indeed be significantly
reduced by optimizing the driving speed.  
Second, our bound \eqref{ApplAQP_NB} underestimates the energy 
uptake by at least a factor of 6 for constant, and at least a factor 
of 4 for optimal driving speed. 
Microscopically, the deviations arise from idle scattering events, 
where carriers pick up energy from the external driving without
contributing to the pump flux. 
This effect becomes more and more dominant as the amplitude $\rho$ of 
the potential modulations increases. 
Still, at least for moderate amplitudes, our bound \eqref{ApplAQP_NB}
predicts the correct order of magnitude for the energy uptake of the 
device. 

\subsection{Adiabatic Quantum Motors}\label{SecAdQuantMot}

\subsubsection{Bound on Efficiency}
A quantum motor can be described in terms of two components: 
a mesoscopic conductor hosting an electric current $J^\r$ between two
reservoirs with the same temperature but different chemical potentials
and a mechanical rotor that couples to the dynamics of the carriers
\cite{Bustos-Marun2013}, see Fig.~\ref{Fig_QMot}. 
Provided that the rotor is much heavier than the carriers, it can be 
treated as a classical degree of freedom, which creates a slowly and
periodically changing potential inside the conductor
\cite{Bode2012a}.
In this Born-Oppenheimer picture, the motive power of the rotor, 
that is, the output of the motor, is given by $\Pi_{{{\rm m}}}\equiv
-\Pi_{{{\rm ac}}}= - T F^\w J^\w$. 
The electric power $\Pi_{{{\rm el}}}=TF^\r J^\r$ is the input of the
motor and its efficiency is defined as 
\begin{equation}\label{ApplAQM_Eff}
\eta_{{{\rm m}}} \equiv \Pi_{{{\rm m}}}/\Pi_{{{\rm el}}}\leq 1
\end{equation}
for $\Pi_{{{\rm m}}}>0$. 
The upper bound $1$ thereby follows from the second law, which 
requires $\sigma = \Pi_{{{\rm el}}}/T-\Pi_{{{\rm m}}}/T \geq 0$.

Going beyond this trivial result, our bound \eqref{ThNB} implies that
the efficiency and the input of the device are connected by the 
universal relation
\begin{equation}\label{ApplAQM_NB} 
\eta_{{{\rm m}}}\leq 1-\frac{\Pi_{{{\rm el}}}}{2T K^{\r\r}(F^\r)^2}
	\leq 1- \frac{h\Pi_{{{\rm el}}}}{\Delta\mu^2}
	= 1- \frac{hJ^\r}{\Delta\mu},
\end{equation}
where the second inequality follows from  Eq.~\eqref{ThNBPropParaa}. 
Depending only on the electric current $J^\r$ and the chemical 
potential bias $\Delta\mu=\mu_1-\mu$, both of which are generally easy
to access in experiments, the bound \eqref{ApplAQM_NB} has a key 
practical implication: 
it makes it possible to put a non-trivial upper limit on the 
efficiency of an adiabatic quantum motor without measuring the motive
power of the rotor or invoking a specific model. 

\subsubsection{Paddle Wheel Motor}\label{Sec_PadWheMot}
\begin{figure}
\includegraphics[scale=1.05]{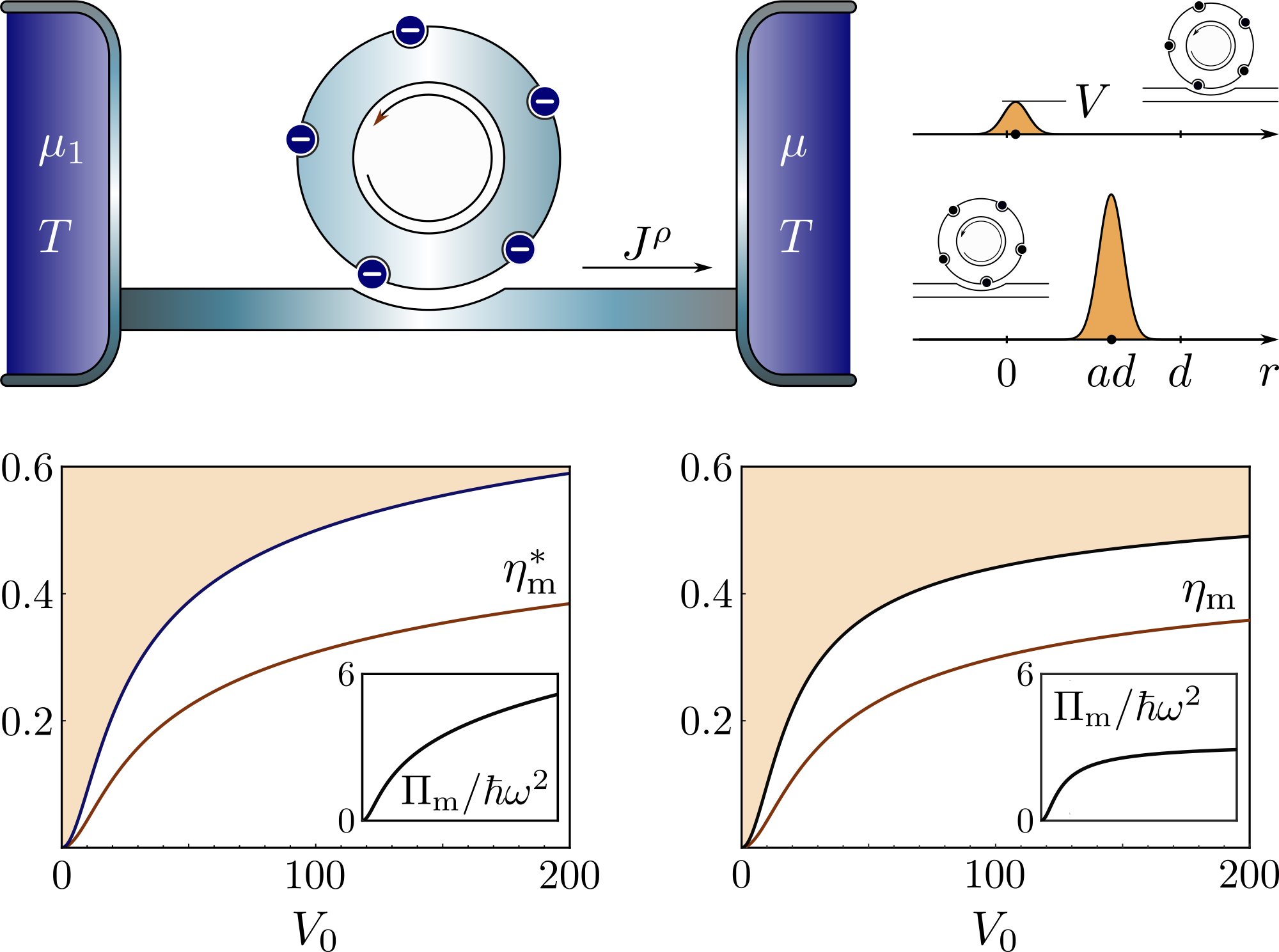}
\caption{Quantum paddle wheel motor. \textbf{Top:} Setup. 
A narrow coherent conductor, which is placed between two isothermal 
reservoirs with chemical potentials $\mu_1$ and $\mu<\mu_1$, is in 
contact with a mesoscopic wheel, the rotor, which carries negative 
charges on its surface acting as electrostatic paddles. 
The electric current $J^\r$ pushes the potential barrier that is 
created by a passing charge to the right thus driving the rotation of
the wheel. 
\textbf{Bottom:}
The left plot shows the maximum efficiency \eqref{ApplAQM_MaxEff} of 
the motor as a function of the  strength of the barrier $V_0$ with the
shaded area corresponding to the bound \eqref{ApplAQM_NB}. 
The inset shows the generated motive power over the same range of 
$V_0$ as in the main panel.
For comparison, the efficiency and the motive power are plotted on the
right for fixed bias $v=10$. 
For all plots, we have set $\chi_\mu=4\pi$. 
\label{Fig_QMot}}
\end{figure}

To test the accuracy of the bound \eqref{ApplAQM_NB}, we consider a 
simple quantum motor, which consists of a rotating nano-paddle wheel 
creating a sliding delta-barrier inside a narrow conductor, see 
Fig.~\ref{Fig_QMot}. 
The single-particle Hamiltonian of this model is
\begin{equation}
H_{\df} = \frac{p^2}{2M} + \frac{\hbar^2 V}{Md} \d_{r-da}
\end{equation}
and the frozen scattering amplitudes are given by 
\begin{subequations}
\begin{align}
\AS^{12}_{E,\df}&=\AS^{21}_{E,\df}=i\chi_E e^{i\chi_E}/(i\chi_E -V),\\
\AS^{11}_{E,\df}&= Ve^{2i a\chi_E}/(i\chi_E -V), \\
\AS^{22}_{E,\df}&= Ve^{2i(1-a)\chi_E}/(i\chi_E -V)
\end{align}
\end{subequations}
with $\df\equiv (V,a)$ and $\chi_E\equiv d\sqrt{2ME}/\hbar$. 
Here, $d$ denotes the length of the contact region between the
conductor and the paddle wheel and the dimensionless control
parameters $V$ and $a\in [0,1]$ correspond to the strength and the 
position of the barrier. 
In the low-temperature limit, the electric and the motive power are 
given by 
\begin{subequations}
\begin{align}
\Pi_{{{\rm el}}}/\hbar\w^2 &= \bar{L}^{\r\r}v^2
	+\bar{L}^{\r\w}v \quad\text{and}\\
\Pi_{{{\rm m}}}/\hbar\w^2  &= -\bar{L}^{w\r}v-\bar{L}^{\w\w},
\end{align}
\end{subequations}
where $v\equiv \Delta\mu/\hw$ and we have introduced the rescaled 
kinetic coefficients 
\begin{subequations}
\begin{align}
\bar{L}^{\r\r} &=
	\frac{1}{2\pi}\bigl\llangle \abs{\AS^{12}_{\mu,\df}}^2\bigr\rrangle\\
\bar{L}^{\r\w} &=-\bar{L}^{\w\r}
	=\frac{1}{2\pi \w}\bsum {{{\rm Im}}}\bigr[\bigr\llangle 
	\AS^{\b 1}_{\mu,\df}\dot{\AS}_{\mu,\df}^{\b 1\ast}
	\bigr\rrangle\bigr],\\
\bar{L}^{\w\w} &= \frac{1}{4\pi\w^2}\absum\bigl\llangle 
	\abs{\dot{\AS}^{\a\b}_{\mu,\df}}^2\bigr\rrangle.
\end{align}
\end{subequations}

We now calculate $\Pi_{{{\rm el}}}$ and $\Pi_{{{\rm m}}}$ for the
protocols 
\begin{equation}
V_t= V_0\sin^2[\w t/2] \quad\text{and}\quad
	a_t = t/\T \;{{{\rm mod}}}\; 1,
\end{equation}
which mimic the motion of the paddle wheel, assuming that its radius
is much larger than $d$. 
The bias $v$ is fixed by maximizing the efficiency of the motor
\eqref{ApplAQM_Eff}. 
That is, we set $v=v^\ast\equiv\bar{L}^{\w\w}\bigl(1+
\sqrt{Z_{{{\rm m}}}+1}\bigr)\bigl/\bar{L}^{\r\w}$, whereby the 
efficiency becomes 
\begin{equation}\label{ApplAQM_MaxEff}
\eta^\ast_{{{\rm m}}}=
	\frac{\sqrt{Z_{{{\rm m}}}+1}-1}{\sqrt{Z_{{{\rm m}}}+1}+1}
\end{equation}
with $Z_{{{\rm m}}}\equiv (\bar{L}^{\r\w})^2/\bar{L}^{\r\r}
\bar{L}^{\w\w}$ playing the role of a figure of merit. 
Our results are plotted in Fig.~\ref{Fig_QMot}. 
For $V_0\lesssim 10$, the bound \eqref{ApplAQM_NB} overestimates the 
efficiency of the paddle-wheel motor by about a factor of $2$. 
As $V_0$ increases, the deviation decays to approximately a factor of
$3/2$, since fewer carriers tunnel through the moving barrier without
transmitting energy to the paddle-wheel.
The same qualitative behavior is observed if the bias $v$ is fixed 
independently of the kinetic coefficients.

\section{Far From Equilibrium}\label{Sec_FFE}
Going beyond the first and the second law, our bounds
\eqref{ThNBCompoundBnd} and \eqref{ThNB} provide strong universal 
constraints on currents in periodically driven coherent conductors. 
They were derived within the framework of linear response theory, 
which describes the arguably most relevant regime of operation
of mesoscopic devices. 
Still, the question remains whether similar bounds apply also far
from equilibrium. 
In this section, we show how such a generalization of our theory 
can be achieved. 
We first derive a family of new bounds on currents in coherent 
conductors that hold for arbitrary driving frequencies and thermal
gradients. 
We then discuss the interpretation of these bounds in the context
of thermodynamic uncertainty relations \cite{Seifert2017a,
Horowitz2020} and show how they can be used for thermodynamic 
inference.
To illustrate the general picture, we revisit the quantum 
generator of Sec.~\ref{SecAdChrPumpTM}, whose Floquet
scattering amplitudes can be calculated exactly for arbitrary driving
frequencies. 
We then put our results in the context of recent developments and 
conclude this section by summarizing the main implications of our 
theory for systems without time dependent driving.

\subsection{Derivation}

To extend the approach of Sec.~\ref{SecThNB} into the non-linear 
regime, we consider the quadratic form 
\begin{equation}\label{FFE_QuadF}
\Xi_\psi \equiv \sigma +\psi\asum\xsum J^x_\a G^x_\a +
	\psi\asum\xysum	\hat{K}^{xy}_\a G^x_\a G^y_\a. 
\end{equation}
Here, $\psi$ is a yet undetermined positive number and the 
coefficients $\hat{K}^{xy}_\a$ are defined as 
\begin{equation}\label{FFE_GenK}
\hat{K}^{xy}_\a \equiv \frac{1}{4h}\Eint 
	\begin{aligned}[t]
		&\bsum\nsum \zeta^{x,\a}_{\En}\zeta^{y,\a}_{\En}
		\abs{\hat{S}^{\a\b}_{\En,E}}^2\\
		&\times \bigl(f^\a_{\En}(1-f^\b_E)+f^\b_E(1-f^\a_{\En})\bigr)
	\end{aligned}
\end{equation}
with $\zeta^{\r,\a}_E\equiv 1$, $\zeta^{q,\a}_E\equiv E-\mu_\a$ and
the modified Floquet scattering amplitudes 
\begin{equation}
\hat{S}^{\a\b}_{\En,E}\equiv (1-\d_{n0}\d_{\a\b})S^{\a\b}_{\En,E}.
\end{equation}
This ansatz is essentially found by inspection but can be 
motivated \emph{a posteriori} as we shall see in 
Sec.~\ref{SecFFE_Int}.

We now recall the formulas \eqref{TwMeanCurrents},
\eqref{TwElPowerHeatCurr} and \eqref{TwEntProd} and use the unitarity
conditions \eqref{TwSumRulesSM} to express the thermo-chemical 
currents and the rate of entropy production as
\begin{subequations}
\begin{align}
J^x_\a & = \frac{1}{h}\Eint \bsum\nsum\zeta^{x,\a}_{\En}
	 \abs{\hat{S}^{\a\b}_{\En,E}}^2 (f^\a_{\En}-f^\b_E),\\
\sigma & = 
	\begin{aligned}[t]
		\frac{1}{h}\Eint &\absum\nsum \abs{\hat{S}^{\a\b}_{\En,E}}^2\\
		&\times \bigl((\nu^\b_E-\nu^\a_{\En})f^\b_E -g^\b_E +g^\a_{\En}\bigr)
	\end{aligned}
\end{align}
\end{subequations}
with $\nu^\a_E\equiv (\mu_\a-E)/T_\a$ and 
$g^\a_E\equiv\nu^\a_E-\ln[f^\a_E]$. 
Inserting these expressions and the definition
\eqref{FFE_GenK} into  Eq.~\eqref{FFE_QuadF} gives 
\begin{equation}\label{FFE_QuadFAux}
\Xi_\psi = \frac{1}{h}\Eint
	\absum\nsum \abs{\hat{S}^{\a\b}_{\En,E}}^2 \Xi^{\a\b}_{\psi,n,E},
\end{equation}
where
\begin{equation}
\Xi^{\a\b}_{\psi,n,E}\equiv
	\begin{aligned}[t]
		&(\nu^\b_E-\nu^\a_{\En})f^\b_E - g^\b_E +  g^\a_{\En}\\
	    &+2\psi X^{\a}_{\En} (f^\a_{\En}-f^\b_E)\\
        &+\psi (X^{\a}_{\En})^2
			\bigl(f^\a_{\En}(1-f^\b_E)+f^\b_E(1-f^\a_{\En})\bigr)
	\end{aligned}
\end{equation}
is a quadratic form in the variables $X^\a_E\equiv \sum_x
\zeta^{x,\a}_E G^x_\a/2$. 
As has been shown in Ref.~\cite{Brandner2017b}, this quadratic form is
positive semidefinite if
\begin{equation}
\psi\leq \psi^\ast\equiv \min_{z\in \mathbb{R}}
	\frac{(1-e^z+ze^z)(e^z+1)}{(e^z-1)^2}\simeq 0.89612. 
\end{equation}
Consequently, we have $\Xi_{\psi^\ast}\geq 0$ for any choice of the 
auxiliary variables $G^x_\a$. 

From here, we can proceed as in Sec.~\ref{SecThNB}. 
Minimizing $\Xi_{\psi^\ast}$, first with respect to the $\{G^\r_\a\}$
and then with respect to the $\{G^q_\a\}$ while setting the
respectively remaining auxiliary variables to zero, yields
the cumulative bound
\begin{equation}\label{FFE_CompBnd}
\sigma \geq \psi^\ast \asum \frac{(J^x_\a)^2}{4\hat{K}^{xx}_\a}, 
\end{equation}
which can be regarded as the counterpart of
Eq.~\eqref{ThNBCompoundBnd}. 
For the individual matter and heat currents, we obtain the bounds
\begin{subequations}\label{FFE_NBDet}
\begin{align}
\label{FFE_NBDetMatr}
\sigma&\geq 
	\frac{\psi^\ast(1+\theta^\r_\a)}{4\hat{K}^{\r\r}_\a}(J^\r_\a)^2
	\quad\text{and}\\
\label{FFE_NBDetHeat}
\sigma&\geq	
	\frac{\psi^\ast}{4\hat{K}^{qq}_\a}\bigl((J^q_\a)^2
	+ (J^q_\a+\Pi_{{{\rm ac}}}+\Pi_{{{\rm el}}})^2\theta^q_\a\bigr)
\end{align}
\end{subequations}
with $\theta^x_\a \equiv\hat{K}^{xx}_\a\bigl/\sum_{\b\neq\a}
\hat{K}^{xx}_\b\geq0$, which can be derived from 
Eq.~\eqref{FFE_CompBnd} by repeating the steps that led from 
Eq.~\eqref{ThNBCompoundBnd} to Eq.~\eqref{ThNB};
for the heat currents, the conservation law $\asum J^q_\a=0$, which
holds only in linear response, must thereby be replaced with the first 
law \eqref{TwFirstLaw}. 

We note that, in linear response, the bounds \eqref{FFE_NBDet} are 
weaker than our previous bound \eqref{ThNB} by a factor $\psi^\ast$.
Specifically, upon neglecting third-order corrections in the 
affinities, Eqs.~\eqref{FFE_NBDet} imply
\begin{equation}
\sigma\geq \frac{\psi^\ast(1+\theta^x_\a|_{{{\rm eq}}})}{
	4\hat{K}^{xx}_\a|_{{{\rm eq}}}}	(J^x_\a)^2
	\geq \frac{\psi^\ast N}{4 K^{xx} (N-1)}(J^x_\a)^2, 
\end{equation}
where the second inequality follows by noting that
\begin{align}
\hat{K}^{xx}_\a|_{{{\rm eq}}}&= K^{xx}- \frac{1}{2h}\Eint\; 
	(\zeta^x_E)^2\abs{\bigl\llangle \AS^{\a\a}_{E,\df}
	\bigr\rrangle}^2 f'_E\\
	&\leq K^{xx}\nonumber
\end{align}
and therefore
\begin{equation}
\frac{1+\theta^x_\a|_{{{\rm eq}}}}{\hat{K}^{xx}_\a|_{{{\rm eq}}}} 
= \frac{1}{\hat{K}^{xx}_\a|_{{{\rm eq}}}} 
	+ \frac{1}{\sum_{\b\neq\a}\hat{K}^{xx}_\b|_{{{\rm eq}}}}
	\geq\frac{N}{K^{xx}(N-1)}. 
\end{equation}
This observation shows that the bounds \eqref{FFE_NBDet} cannot
be saturated close to equilibrium.

\subsection{Thermodynamic Uncertainty Relations}\label{SecFFE_Int}

A thermodynamic uncertainty relation describes the trade-off between
precision and entropy production in a given class of thermodynamic 
processes \cite{Seifert2017a,Horowitz2020}. 
For classical matter transport in multi-terminal systems
without time-dependent driving, for example, the relation
\begin{equation}\label{FFED_TURCl}
\sigma (\epsilon^\rho_\a)^2 \geq \psi^\ast
\end{equation}
has been derived in Ref.~\cite{Brandner2017b}. 
The reduced zero-frequency noise, or squared relative uncertainty,
$(\epsilon^\rho_\a)^2 \equiv P^{\r\r}_{\a\a}/(J^\r_\a)^2$,
thereby provides a measure for the accuracy at which a given amount of
particles is extracted from the reservoir $\a$ in a given time. 
As we show in the following, our bounds \eqref{FFE_NBDet} make it 
possible to extend this relation into the quantum regime and
to include heat currents as well as periodic driving. 

We first recall that, according to the formulas \eqref{TwMEN} and
\eqref{TwMENDet}, the diagonal thermal and the shot noise of the 
thermo-chemical currents are given by 
\begin{subequations}
\begin{align}
\label{FFED_ThNExpl}
D^{xx}_{\a\a}&= \sum\nolimits_{uv} c^{xu}_\a c^{xv}_\a D^{uv}_{\a\a}\\
	&\begin{aligned}[t]
	=\frac{1}{h}\Eint \bsum\nsum\big(
	(&\zeta^{x,\a}_{\En})^2 \abs{S^{\a\b}_{\En,E}}^2(f'^\a_{\En} + f'^\b_E)\\
	&-2\d_{\a\b}\zeta^{x,\a}_{\En}\zeta^{x,\a}_E\abs{S^{\a\a}_{\En,E}}^2f'^\a_E
	\big),
	\end{aligned}
	\nonumber\\[3pt]
\label{FFED_ShNExpl}
R^{xx}_{\a\a} &=\sum\nolimits_{uv} c^{xu}_\a c^{xv}_\a R^{uv}_{\a\a}\\
	&= \frac{1}{2h}\Eint \sum\nolimits_{\g\d}\nsum 
	\Bigl|\sum\nolimits_u c^{xu}_\a C^{u\a,\g\d}_{n,E}\Bigr|^2
	\nonumber
\end{align}
\end{subequations}
with $c^{\r u}_\a\equiv\d_{u\r}$, $c^{q u}_\a\equiv\d_{u\ve}
-\mu_\a\d_{u\r}$ and $f'^\a_E\equiv f_E^\a(1-f^\a_E)$.
Next, comparing Eq.~\eqref{FFED_ThNExpl} with Eq.~\eqref{FFE_GenK} 
shows that the coefficients $\hat{K}^{xx}_\a$ can be decomposed as
\begin{equation}\label{FFED_KDecomp}
4\hat{K}^{xx}_\a = D^{xx}_{\a\a} + \Psi^{xx}_\a + \Omega^{xx}_\a. 
\end{equation}
The Fermi correction 
\begin{equation}\label{FFE_PauliCorr}
\Psi^{xx}_\a\equiv \frac{1}{h}\Eint \bsum\nsum (\zeta^{x,\a}_{\En})^2
	\abs{\hat{S}^{\a\b}_{\En,E}}^2(f^\a_{\En}-f^\b_E)^2
\end{equation}
thereby accounts for Pauli blocking between incoming and outgoing 
carriers; 
it vanishes in equilibrium and in the quasi-classical limit, where 
second-order terms in the fugacities $\varphi_\a\equiv 
\exp[\mu_\a/T_\a]$ can be neglected and the exclusion principle 
becomes irrelevant \cite{Callen1985}. 
The second correction in Eq.~\eqref{FFED_KDecomp}, 
\begin{equation}\label{FFE_DrCorr}
\Omega^{xx}_\a\equiv \frac{2}{h}\Eint
	\nsum \zeta^{x,\a}_{\En}\zeta^{x,\a}_E
	\abs{\hat{S}^{\a\a}_{\En,E}}^2f'^\a_E, 
\end{equation}
arises from inelastic reflections of carriers at the sample and 
vanishes if the external driving is turned off, i.e., if 
$S^{\a\a}_{\En,E} = \d_{n0} S^{\a\a}_E$ and thus $\hat{S}^{\a\a}_{\En, 
E}=0$. 
Finally, since the shot noise \eqref{FFED_ShNExpl} is non-negative, 
Eq.~\eqref{FFED_KDecomp} implies that $\hat{K}^{xx}_\a\leq P^{xx}_{
\a\a}+\Psi^{xx}_\a +\Omega^{xx}_\a$, where $P^{xx}_{\a\a}=D^{xx}_{
\a\a}+R^{xx}_{\a\a}$ is the zero-frequency noise of the 
thermo-chemical currents. 
Inserting this bound into the Eqs.~\eqref{FFE_NBDet} and setting 
$\theta^x_\a=0$ yields the thermodynamic uncertainty relation
\begin{equation}\label{FFED_TUR}
\sigma (\epsilon^x_\a)^2\geq \frac{\psi^\ast}{
1+\Psi^{xx}_\a/P^{xx}_{\a\a}+\Omega^{xx}_\a/P^{xx}_{\a\a}},
\end{equation}
where $(\epsilon^x_\a)^2\equiv P^{xx}_{\a\a}/(J^x_\a)^2$. 

This result can be interpreted as follows.
The corrections $\Psi^{\r\r}_\a$ and $\Omega^{\r\r}_\a$
are non-negative. 
Therefore, the relation \eqref{FFED_TUR} shows that Pauli 
blocking and periodic driving are both sources of precision, which 
reduce the minimal amount of entropy that must be produced to 
generate an arbitrary matter current $J^\r_\a$ with given 
uncertainty $\epsilon^\r_\a$. 
These two mechanisms of reducing the thermodynamic cost of 
precision have been described separately before for
different classes of systems, see for instance 
Refs.~\cite{Brandner2017b,Agarwalla2018,Liu2019a} and
\cite{Barato2016,Barato2018b,Holubec2018,Macieszczak2018,Dechant2019a,
Koyuk2019,Koyuk2019b,VanVu2019,Dechant2020}, respectively. 
Our uncertainty relation \eqref{FFED_TUR} now 
accounts for both of them in a unified manner through additive 
corrections. 

For heat currents, the situation is slightly different.
The corresponding Pauli corrections $\Psi^{qq}_\a$ are non-negative 
and therefore universally suppress the lower bound on the product
$\sigma(\epsilon^q_\a)^2$ compared to the quasi-classical 
limit, where $\Psi^{qq}_\a=0$;
this observation is in line with earlier results, see
Ref.~\cite{Saryal2019}.
The sign of the driving corrections $\Omega^{qq}_\a$, however, depends
on the structure of the Floquet scattering amplitudes of the sample. 
Periodic driving can thus either reduce or increase the minimal 
thermodynamic cost of heat transport at given precision. 

We note that the bound \eqref{FFED_TURCl}, to which the 
relation \eqref{FFED_TUR} reduces in the quasi-classical limit and 
without periodic driving, can be asymptotically saturated in systems 
with infinitely many terminals and strong biases \cite{Brandner2017b}.
Hence, for matter currents, also the uncertainty relation 
\eqref{FFED_TUR} and the stronger bound \eqref{FFE_NBDetMatr}
can be regarded as tight.
Whether these bounds are also tight for heat currents is an open
question.

\subsection{Thermodynamic Inference}
Although the corrections $\Psi^{xx}_\a$ and $\Omega^{xx}_\a$ admit a 
transparent physical interpretation, it is not clear how they can be
measured. 
Before the generalized uncertainty relation \eqref{FFED_TUR} can be 
used for thermodynamic inference, it is therefore necessary to link
these quantities to experimentally accessible observables. 
For matter currents, such a connection is provided by the bounds
\begin{equation}\label{FFE_TIAux}
\Psi^{\r\r}_\a \leq \sigma/2 
	\quad\text{and}\quad
\Omega^{\r\r}_\a \leq 2T_\a/h
\end{equation}
which ultimately follow from the unitarity conditions 
\eqref{TwSumRulesSM}, see Lemmas~\ref{Lem4} and \ref{Lem5} of
App.~\ref{App_Lemmas} for details. 
Inserting these bounds into Eq.~\eqref{FFED_TUR} yields 
\begin{equation}\label{FFE_TI}
\frac{\sigma(P^{\r\r}_{\a\a}+\sigma/2+2T_\a/h)}{(J^\r_\a)^2}\geq
	\psi^\ast. 
\end{equation}

This relation provides a powerful tool for thermodynamic inference.
Specifically, upon measuring an arbitrary matter current $J^\r_\a$,
its zero-frequency noise $P^{\r\r}_{\a\a}$ and the temperature 
$T_\a$ of the corresponding reservoir, one obtains the lower bound 
\begin{equation}\label{FFE_TIsigma}
\sigma\geq \sigma^\ast_\a 
\equiv \abs{J^\r_\a}\left(\sqrt{\phi_\a^2 +2\psi^\ast}-|\phi_\a|
\right)
\end{equation}
on the dissipation rate $\sigma$, where 
$\phi_\a\equiv (hP^{\r\r}_{\a\a}+2T_\a)/h J^\r_\a$ is a 
modified Fano factor. 
This bound holds, within the limits of the scattering approach to 
quantum transport, for any sample and driving protocols, arbitrary 
driving frequencies, arbitrary voltage and temperature biases and in
presence of magnetic fields. 
It therefore makes it possible to derive universal constraints on 
otherwise unaccessible figures of merit of coherent transport devices. 
For instance, upon recalling the setups of Secs.~\ref{SecAdQuantMot}
and \ref{SecAdChrPumpTBP}, the energy uptake $U$ of a parametric 
quantum pump and the efficiency of a quantum motor $\eta_{{{\rm m}}}$
can be bounded as 
\begin{subequations}
\begin{align}
U & \geq T\T\abs{J^\r}\left(\sqrt{\phi^2+2\psi^\ast}-|\phi|\right)
\quad\text{and}\\
\eta_{{{\rm m}}} &\leq 1- \frac{T}{\Delta\mu} 
	\Bigl(\sqrt{\phi^2 +2\psi^\ast}-|\phi|\Bigr)
\end{align}
\end{subequations}
where the parameter $\phi\equiv (hP^{\r\r}_{11}+2T)/J^\r$ is
experimentally accessible. 
These relations can be regarded as far-from-equilibrium 
counterparts of our linear-response bounds \eqref{ApplAQP_NB} and
\eqref{ApplAQM_NB}.

\subsection{Quantum Generator Revisited}\label{Sec_FFEQGR}
\begin{figure}
\includegraphics[scale=1.05]{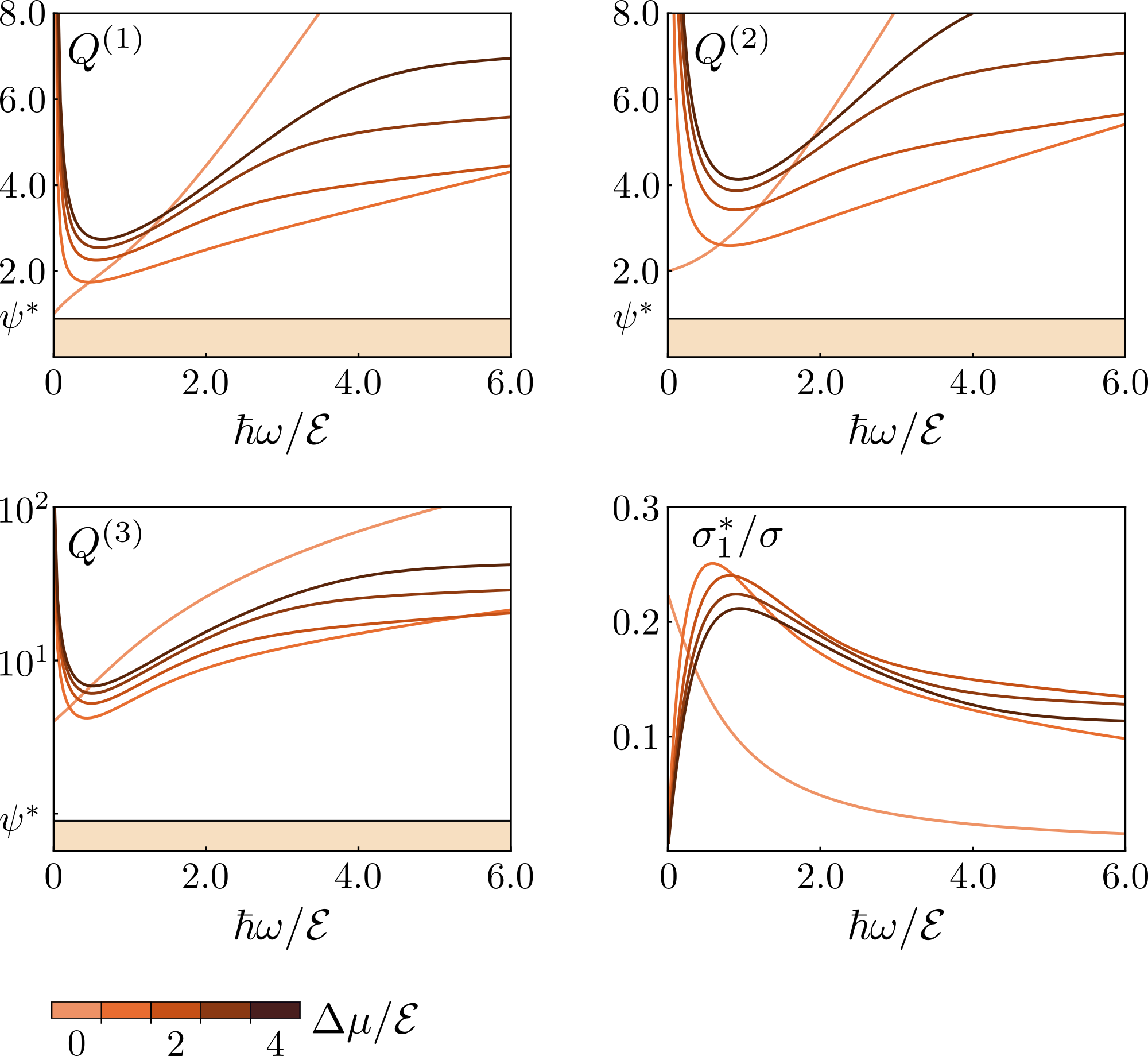}
\caption{
Quantum Generator far from Equilibrium.
The first three plots from top left to bottom left show the quantities
\eqref{FFEQJ_BFig1}-\eqref{FFEQJ_BFig3} as functions of the driving 
frequency for different biases. 
Here, $\mathcal{E}\equiv \hbar^2/2Ml^2$ denotes the natural energy
scale of the device with $M$ being the carrier mass and $l$ the 
circumference of the loop, cf. Fig.~\ref{Fig_QuantGen}. 
The shaded areas indicate the bounds 
\eqref{FFE_NBDetMatr}, \eqref{FFED_TUR} and \eqref{FFE_TI}, 
respectively for $Q^{(1)}$, $Q^{(2)}$, $Q^{(3)}$. 
The last plot shows the lower bound on entropy production 
$\eqref{FFEQJ_BFig4}$ in units of the actual dissipation rate
$\sigma$.
For all plots, we have set $\mu=0$, $\mu_1/\mathcal{E}=\Delta\mu/
\mathcal{E} =0,1,\dots,4$ and $T/\mathcal{E}=1/2$. 
\label{Fig_QGRev}}
\end{figure}

To illustrate our far-from-equilibrium theory, we now return to
the quantum generator of Sec.~\ref{SecAdChrPumpTM}. 
The Floquet scattering amplitudes for this system can be found 
exactly by solving the corresponding Schr\"odinger equation, 
see App.~\ref{Apx_QGen}. 
It is thus straightforward to numerically calculate the quantities
\begin{subequations}\label{FFEQJ_BFig}
\begin{align}
\label{FFEQJ_BFig1}
& Q^{(1)} \equiv
	\frac{4\sigma\hat{K}^{\r\r}_1}{(1+\theta^\r_1)(J^\r_1)^2},\\[3pt]
\label{FFEQJ_BFig2}
& Q^{(2)} \equiv 
	\sigma(\epsilon^\r_1)^2(1+\Psi^{\r\r}_1/P^{\r\r}_{11}
	+\Omega^{\r\r}_1/P^{\r\r}_{11}),\\[3pt]
\label{FFEQJ_BFig3}
& Q^{(3)} \equiv 
	\frac{\sigma (P^{\r\r}_{11} + \sigma/2 +2T/h)}{(J^\r_1)^2}
	\quad\text{and}\\[3pt]
\label{FFEQJ_BFig4}
& \sigma^\ast_1 = \abs{J^\r_1}\left(\sqrt{\phi_1^2+2\psi^\ast}-\phi_1\right)
\end{align}
\end{subequations}
which are plotted Fig.~\ref{Fig_QGRev}.

According to the relations \eqref{FFE_NBDet}, \eqref{FFED_TUR}
and \eqref{FFE_TI}, the dimensionless coefficients
$Q^{(1)}$, $Q^{(2)}$ and $Q^{(3)}$ are bounded from below by 
$\psi^\ast$. 
The coefficient $Q^{(1)}$ indeed comes close to this bound for 
small biases and driving frequencies. 
In fact, we have $Q^{(1)}\rightarrow 1$ for $\Delta\mu,\w
\rightarrow 0$, which confirms our previous observation that the 
linear-response counterpart \eqref{ThNB} of the bound 
\eqref{FFE_NBDet} is saturated for the quantum generator.
After passing through its minimum, $Q^{(1)}$ grows monotonically
with the driving frequency. 
Hence, the bound \eqref{FFE_NBDet} becomes less and less tight when 
the system is driven faster. 
The coefficients $Q^{(2)}$ and $Q^{(3)}$ show a qualitatively 
similar dependence on the driving frequency. 
However, their minimums do not become smaller than $2$ and $4$, 
respectively.
Finally, we find that the figure $\sigma^\ast_1$, at best,
underestimates the actual dissipation rate by about a factor $4$. 
This result proves that the bound \eqref{FFE_TI} is generally
strong enough to predict the correct order of magnitude for $\sigma$.

\subsection{Comparison with Earlier Results}
\begin{figure}
\includegraphics[scale=1.05]{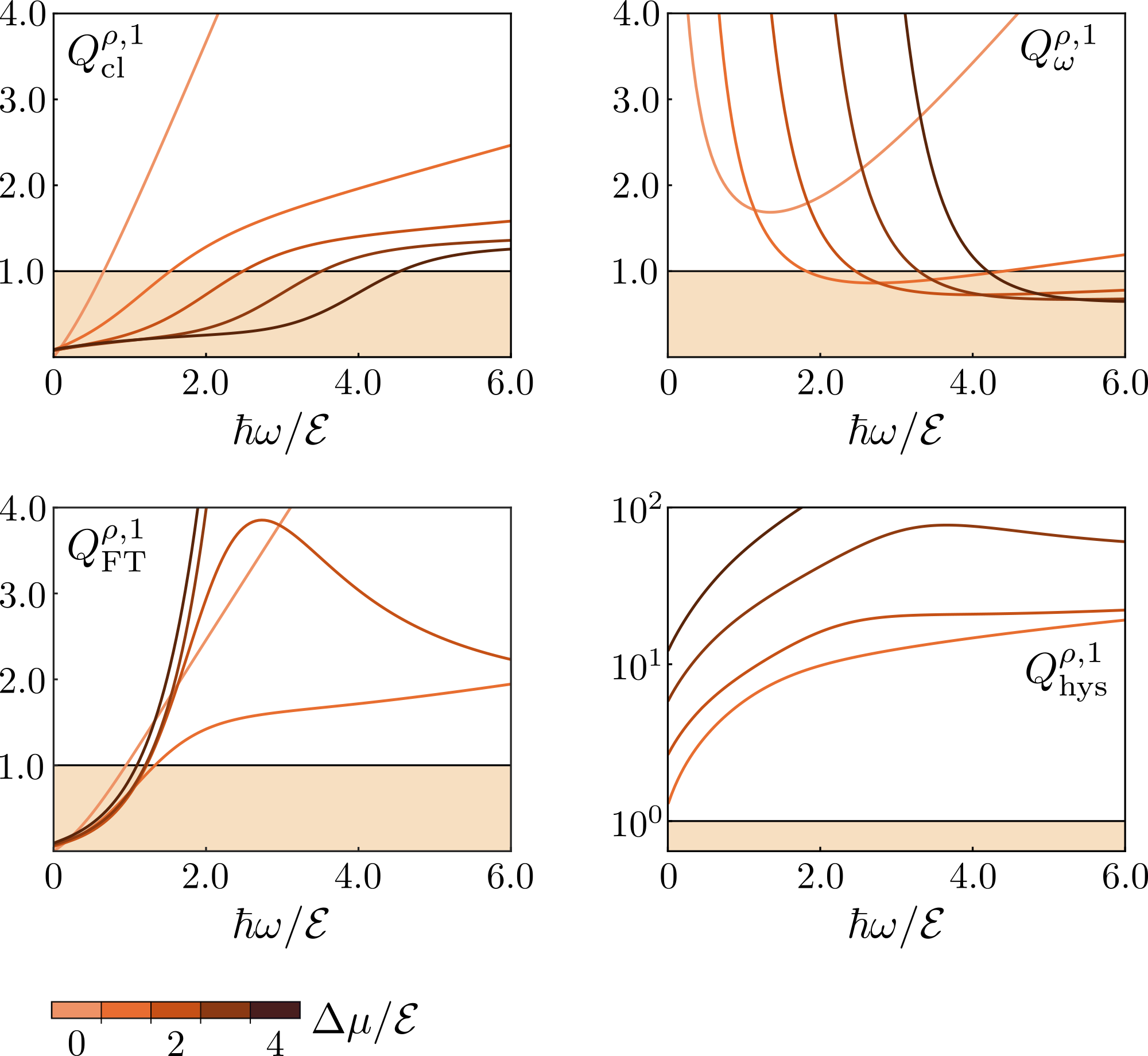}
\caption{
Quantum Generator in Perspective. 
The four plots from top left to bottom right show the coefficients
\eqref{FFE_ER1}-\eqref{FFE_ER4} for $x=\r$ and $\a=1$ as functions of
the driving frequency for $T/\mathcal{E}=1/2$, $\mu=0$ and different 
biases $\Delta\mu/\mathcal{E}=0,1,\dots,4$, where $\mathcal{E}$ was 
defined in Fig.~\ref{Fig_QGRev}. 
In each plot, the shaded area indicates the putative lower bound $1$ 
on the corresponding coefficient. 
\label{Fig_EarlRes}}
\end{figure}

Thermodynamic uncertainty relations are currently a subject of active
research in stochastic thermodynamics, see 
Refs.~\cite{Seifert2017a,Horowitz2020} for recent reviews. 
To put our bounds in context with these developments, we compare 
them with the four relations 
\begin{subequations}\label{FFE_ER}
\begin{align}
\label{FFE_ER1}
Q_{{{\rm cl}}}^{\r,\a} &\equiv \frac{\sigma P^{\r\r}_{\a\a}}{\psi^\ast (J^\r_\a)^2}\geq 1,\\[3pt]
\label{FFE_ER2}
Q_\w^{x,\a} &\equiv \frac{\sigma P^{xx}_{\a\a}}{2(J^x_\a -\w\partial_\w J^x_\a)^2}\geq 1,\\[3pt]
\label{FFE_ER3}
Q_{{{\rm FT}}}^{x,\a} &\equiv \frac{(\exp[\T\sigma]-1)P^{xx}_{\a\a}}{2\T (J^x_\a)^2}\geq 1
	\quad\text{and}\\[3pt]
\label{FFE_ER4}
Q_{{{\rm hys}}}^{x,\a} &\equiv \frac{(\exp[\T(\sigma+\tilde{\sigma})/2]-1)
	(P^{xx}_{\a\a}+\tilde{P}^{xx}_{\a\a})}{(J^x_\a +\tilde{J}^x_\a)^2}\geq 1,
\end{align}
\end{subequations}
which were derived earlier for different classes of systems 
\cite{Brandner2017b,Koyuk2019b,Hasegawa2019a,Proesmans2019a,
Potts2019,Falasco2020}.
For concreteness, we focus again on the quantum generator of
Sec.~\ref{SecAdChrPumpTM}, for which the coefficients 
$Q_{{{\rm cl}}}^{\r,1}$, $Q_\w^{\r,1}$, $Q_{{{\rm FT}}}^{\r,1}$ and
$Q_{{{\rm hys}}}^{\r,1}$ can be easily calculated.
The results of this analysis, which are summarized in 
Fig.~\ref{Fig_EarlRes}, lead to the following conclusions.

First, the bound \eqref{FFE_ER1}, which was derived for classical
ballistic transport without periodic driving \cite{Brandner2017b}, is
violated at low frequencies. 
In fact, the coefficient $Q_{{{\rm cl}}}^{\r 1}$ goes to zero for  
$\Delta\mu,\w\rightarrow 0$. 
This behavior can be understood by observing that the driving 
corrections \eqref{FFE_DrCorr} do, in contrast to the Pauli 
corrections \eqref{FFE_PauliCorr}, not vanish in adiabatic 
equilibrium.  
Instead, we have 
\begin{equation}
\Omega^{xx}_\a|_{{{\rm eq}}} = \frac{2}{h}\Eint \;
	(\zeta^x_E)^2\Bigl(\bigl\llangle \abs{\AS^{\a\a}_{E,\df}}^2\bigr\rrangle 
	-\abs{\bigl\llangle \AS^{\a\a}_{E,\df}\bigr\rrangle}^2\Bigr)f'_E
\end{equation}
and thus $\Omega^{\r\r}_1|_{{{\rm eq}}}=2T\varphi/h(1+\varphi)>0$ with
$\varphi=\exp[\mu/T]$ the quantum generator. 
We recall that the notation $\cdots|_{{{\rm eq}}}$ indicates the limit
$F^x_\a,\w\rightarrow 0$. 
In general, the corrections $\Omega^{xx}_\a$ vanish only if the 
driving fields are switched off, that is, if their amplitudes rather 
than their frequency go to zero. 
In this limit, the quantum generator does not produce any current and 
the bound \eqref{FFE_ER1} becomes trivial. 

The second relation \eqref{FFE_ER2} was originally derived for 
periodically modulated Markov jump processes \cite{Koyuk2019b}. 
In linear response, it holds also for coherent transport, provided
that the frozen scattering amplitudes obey the symmetry
$\AS^{\a\b}_{E,\df}=\AS^{\b\a}_{E,\df}$, as is the
case for the quantum generator, cf. Eq.~\eqref{FrABScatAmpl}; 
for details, see App.~\ref{Apx_TURLR}. 
Far from equilibrium, however, the bound \eqref{FFE_ER2} can be 
violated as the second plot in Fig.~\ref{Fig_EarlRes} proves. 

The \emph{fluctuation theorem uncertainty relation} \eqref{FFE_ER3}
follows from a general symmetry argument. 
Specifically, it holds if the joint probability distribution of the
entropy production per cycle $\T\sigma$ and the accumulated current
$\T J^x_\a$ obeys a strong detailed 
fluctuation theorem \cite{Hasegawa2019a}. 
In general, such a relation holds only if no magnetic fields are 
applied to the system and the driving protocols are symmetric under 
time reversal. 
These conditions do not apply to the quantum generator. 
As a result, the third plot in Fig.~\ref{Fig_EarlRes} shows clear 
violations of the bound \eqref{FFE_ER3} at low frequencies. 

The \emph{hysteretic uncertainty relation} \eqref{FFE_ER4} provides a
generalization of Eq.~\eqref{FFE_ER3}. 
The restriction to systems without magnetic fields and symmetric 
driving protocols is thereby removed by considering the actual
thermodynamic process together with its time-reversed counterpart, 
which is denoted by a tilde \cite{Proesmans2019a,Potts2019};
in coherent transport, quantities with a tilde are obtained from 
original ones by replacing $S^{\a\b}_{\En,E}$ with 
$\mathsf{T}_\mathbf{B}\mathsf{T}_{\df}S^{\a\b}_{\En,E}
=S^{\b\a}_{E,\En}$. 
The last plot in Fig.~\ref{Fig_EarlRes} shows that the hysteretic 
uncertainty relation \eqref{FFE_ER4} is not violated for the quantum
generator. 
Whether or not it holds for coherent transport in general remains an
open question at this point. 
In any case, the relation \eqref{FFE_ER4} has only limited predictive
power in practice. 
First, it is inapplicable if the time-reversed process cannot be 
realized. 
Second, since it involves only the symmeterized variable $\sigma+
\tilde{\sigma}$, it cannot be used to bound the entropy production
$\sigma$ of the actual transport process.  
Third, if $\tilde{J}^x_\a = -J^x_\a$, as is the case for $J^\r$ in 
the setup of the quantum generator, Eq.~\eqref{FFE_ER4} reduces to the 
trivial bound $\sigma+\tilde{\sigma}\geq 0$. 
By contrast, our relation \eqref{FFE_TIsigma} still yields a 
non-trivial lower bound on $\sigma$ in this situation. 

In summary, our case study shows that the new bounds
\eqref{FFE_NBDet}, \eqref{FFED_TUR} and \eqref{FFE_TI}, which fully 
incorporate quantum effects and do not rely on special symmetries, go
beyond the earlier results \eqref{FFE_ER}. 
In particular, Eq.~\eqref{FFE_TI} provides a universal bound on 
entropy production that depends only on parameters that can be 
measured in a given setup without having to realize the 
time-reversed transport process. 

\subsection{Autonomous Systems}

In the previous section, we have considered our 
far-from-equilibrium bounds on coherent transport in the context of
earlier results for periodically driven thermodynamic processes. 
However, our theory also has profound implications for 
autonomous systems, which can be established as follows. 
If no time dependent fields are applied to the sample, the driving
correction \eqref{FFE_DrCorr} is zero. 
The bounds \eqref{FFED_TUR} and \eqref{FFE_TIAux} then imply 
the new relation
\begin{equation}\label{FFE_ATI}
\frac{\sigma(P^{\r\r}_{\a\a}+\sigma/2)}{(J^\r_\a)^2} 
	= \sigma (\epsilon^\r_\a)^2 + \frac{\sigma^2}{2(J^\r_\a)^2} 
	\geq \psi^\ast. 
\end{equation}
This result shows that Pauli-blocking as a source of precision can be 
incorporated into the classical uncertainty relation 
\eqref{FFED_TURCl}, which otherwise can be violated in the quantum 
regime \cite{Brandner2017b,Agarwalla2018,Liu2019a}, through the 
universal correction $\sigma^2/2(J^\r_\a)^2$. 
Quite remarkably, this \emph{quantum shift} does not depend on any 
additional parameters. 
As a result, the relation \eqref{FFE_ATI} leads to a bound on entropy 
production, 
\begin{equation}\label{FFE_ATIsigma}
\sigma \geq \abs{J^\r_\a}
	\left(\sqrt{\mathrm{F}_\a^2 +2\psi^\ast}-|\mathrm{F}_\a|\right),
\end{equation}
that involves only the current $J^\r_\a$ and the standard
zero-frequency Fano factor $\mathrm{F}_\a\equiv
P^{\r\r}_{\a\a}/J^\r_\a$. 
Like the relation \eqref{FFE_TIsigma}, this bound holds for any 
coherent multi-terminal conductor, arbitrary electric and thermal 
biases and in presence of external magnetic fields. 

Finally, we note that, because the bound \eqref{FFE_TIAux} on the 
driving corrections $\Omega^{\r\r}_\a$ is independent of the driving
strength, the relations \eqref{FFE_ATI} and \eqref{FFE_ATIsigma}
cannot be recovered from their more general counterparts 
\eqref{FFE_TI} and \eqref{FFE_TIsigma}. 
This observation may indicate that the bounds \eqref{FFE_TI} and
\eqref{FFE_TIsigma} can be further optimized. 
We leave it to future research to probe whether such refined 
bounds exist. 
Furthermore, it remains yet an open problem whether operationally 
accessible bounds on entropy production, similar to the ones given 
in Eqs.~\eqref{FFE_TIsigma} and \eqref{FFE_ATIsigma}, can be 
formulated in terms of individual heat currents. 

\section{Summary}\label{SecConclusion}
Starting from the scattering approach to quantum transport, we have 
developed a universal thermodynamic framework for coherent 
conductors that are driven by thermo-chemical gradients and 
periodically changing electromagnetic fields, whose frequency plays
the role of an additional thermodynamic force. 
Focusing on the linear-response regime, we have shown that this 
framework can be equipped with consistent generalizations of the
Onsager-Casimir relations and the fluctuation-dissipation theorem. 
As our first key result, we have derived a family of 
thermodynamic bounds on matter and heat currents, which go beyond the
second law, hold for arbitrary samples and driving protocols and 
involve only experimentally accessible quantities. 

From a conceptual point of view, these bounds prove that transport 
without dissipation is impossible in conventional coherent conductors,
even in the presence of reversible currents, which generically occur 
in systems with broken time reversal symmetry and do not contribute to
the average entropy production. 
From a practical perspective, they provide powerful tools of 
thermodynamic inference. 
In particular, they make it possible to determine model-independent
lower bounds on the total dissipation rate, which is generally 
difficult to access experimentally, by measuring the electric currents
in the individual terminals of the conductor. 

When applied to mesoscopic devices, our bounds lead to non-trivial
relations between key figures of merit, which provide both universal
benchmarks for theoretical models and a new avenue to estimate the
performance of experimental realizations of nano-machines, whose 
output or input cannot be measured directly. 
We have illustrated this method for parametric quantum pumps and 
adiabatic quantum motors, where in both cases a mechanical quantity 
was bounded in terms of an easy-to-measure electric current.
Beyond these examples, our results are applicable to any system that
can be described as a coherent multi-terminal conductor, including
thermoelectric heat engines and refrigerators, which have gained much 
attention in recent years
\cite{Avron2001,Benenti2011,Brandner2013,Whitney2013,Brandner2013a,
Mazza2014,Whitney2014,Sothmann2014a,Brandner2015,Whitney2014a,
Samuelsson2017,Brandner2017b,Macieszczak2018,Luo2018}. 
Our work thus provides a versatile toolbox for both theoretical and 
experimental studies seeking to develop powerful and efficient quantum
transport devices.  

Going beyond the linear-response regime, in the second part 
of this article, we have derived a family of thermodynamic bounds on
matter and heat currents that hold for arbitrary thermo-chemical 
gradients and driving frequencies. 
These relations imply a thermodynamic uncertainty relation for 
coherent transport, which accounts for quantum effects and periodic
driving in a unified manner, and a fully universal bound on entropy 
production that depends only on experimentally accessible parameters.
This bound, which is our second key result, can be directly used for
thermodynamic inference far from equilibrium. 
In fact, it can be determined by measuring a single electric current,
its zero-frequency noise and, for systems with periodic driving, the
temperature of the corresponding reservoir. 
Thus, we are now able to pass the baton to the experimentalists to 
test our theoretical results and to utilize their wide-ranging 
applicability to further explore the quantum thermodynamics of 
coherent transport devices. 

\begin{acknowledgments}
E.P. acknowledges helpful comments from V. Kashcheyevs. 
K.B. thanks K. Saito for insightful discussions. 
E.P. acknowledges support from the Vilho, Yrj\"o and Kalle Foundation
of the Finnish Academy of Science and Letters through the grant for 
doctoral studies.
C.F. acknowledges support from the Academy of Finland 
(Projects No. 308515 and No. 312299).
M.M. acknowledges the warm hospitality of Aalto University, support 
from the Aalto Science Institute through its Visiting Fellow 
Programme, and support from the Ministry of Education and Science of 
Ukraine (project No. 0119U002565). 
K.B. acknowledges support from Academy of Finland
(Contract No. 296073),  the Japan Society for the Promotion of 
Science through a Postdoctoral Fellowship for Research in Japan
(Fellowship  ID:  P19026), the University of Nottingham through a 
Nottingham Research Fellowship and from  UK  Research  and  Innovation
through  a  Future Leaders Fellowship (Grant Reference: MR/S034714/1).
Authors from Aalto University are associated with the local Centre
for Quantum Engineering.
\end{acknowledgments}
\vspace*{1cm}

\appendix
\newcommand{\de}{\partial_E^{\phantom{\b}}}
\newcommand{\dft}{\df\ix{t}}
\newcommand{\dAS}{\dot{\AS}}
\newcommand{\av}[1]{\bigl\llangle #1\bigr\rrangle}
\newcommand{\tr}[1]{{{{\rm tr}}}\bigl\{#1\bigr\}}
\newcommand{\Sb}{\mathbb{S}}
\newcommand{\Ab}{\mathbb{A}}
\newcommand{\pd}{\phantom{\dagger}}

\section{Kinetic Coefficients}\label{Appx_KinCoeff}

In this appendix, we derive the formulas \eqref{TwKinCoeffExp} for 
the linear-response coefficients in the frequency picture, that is,
with the driving frequency playing the role of an additional affinity.
We then compare this approach with the amplitude picture, where
the mechanical affinity corresponds to the strength of the applied
periodic fields.  

\subsection{Frequency Picture}\label{Appx_KinCoeff_Ad}
We first recall that the thermo-chemical currents and the photon
current are given by 
\begin{subequations}
\begin{align}
&J^x_\a = \frac{1}{h}\Eint \bsum\nsum\Bigl(
	\zeta^{x,\a}_E\d_{\a\b}\d_{n0}
	-\zeta^{x,\a}_{\En}\abs{S^{\a\b}_{\En,E}}^2\Bigr)f^\b_E
	\nonumber\\[3pt]
\label{AppAKinCurr}
&\text{with}\quad\begin{aligned}[t]
		\zeta^{x,\a}_E &=\zeta^x_E-\d_{xq}TF^\r_\a,\\[3pt]
		\zeta^{x,\a}_E &=\zeta^x_E-\d_{xq}T(F^\r_\a-nF^\w)
		\quad\text{and}
	\end{aligned}\\[3pt]
& J^\w =\frac{1}{h}\Eint\absum\nsum n \abs{S^{\a\b}_{\En,E}}^2 f^\b_E,
\end{align}
\end{subequations}
where $\zeta^\r_E\equiv 1$, $\zeta^q_E\equiv E-\mu$ and 
$\zeta^\w=1/\w$.
Inserting the expansion of the Fermi function \eqref{TwFermiFunctExp}
and the expansions 
\begin{subequations}
\begin{align}
&\nsum \abs{S^{\a\b}_{\En,E}}^2 = \mathcal{X}_0 
	+ \frac{TF^\w}{2}\mathcal{X}_1,\\[3pt]
\label{AppAExpTCoefb}
&\nsum n\abs{S^{\a\b}_{\En,E}}^2 = \mathcal{Y}_0 
	+ \frac{TF^\w}{2}\mathcal{Y}_1
\end{align}
\end{subequations}
and collecting first-order contributions in the affinities then yields
$F^x_\a$ and $F^\w$ yields 
\begin{subequations}\label{AppAKinCoef1Aux}
\begin{align}
\label{AppAKinCoef1Aux1}
L^{xy}_{\a\b} &= \frac{1}{h}\Eint\;\zeta^x_E\zeta^y_E\big(
	\d_{\a\b}-\mathcal{X}_0\big)f'_E,\\[3pt]
\label{AppAKinCoef2Aux1}
L^{x\w}_\a &= -\frac{T}{2h}\Eint \bsum\bigl(2\mathcal{Y}_0\d_{xq}
	+\zeta^x_E\mathcal{X}_1\big)f_E,\\[3pt]
\label{AppAKinCoef3Aux1}
L^{\w x}_\b &=\frac{1}{h}\Eint \; \zeta^x_E \asum 
	\mathcal{Y}_0 f'_E,\\[3pt]
\label{AppAKinCoef4Aux1}
L^{\w\w} &= \frac{T}{2h}\Eint \; \absum \mathcal{Y}_1 f_E.
\end{align}
\end{subequations}
Note that the zeroth-order terms vanish as can be shown with the help
of the unitarity conditions \eqref{TwARSumRulesSM} for the frozen
scattering amplitudes. 

Upon inserting the expressions \eqref{Lem_1a}-\eqref{Lem_1d} for 
$\mathcal{X}_0,\mathcal{X}_1,\mathcal{Y}_0$ and $\mathcal{Y}_1$ in
the Eqs.~\eqref{AppAKinCoef1Aux}, we arrive at
\begin{widetext}
\begin{subequations}\label{AppAKinCoefAux2}
\begin{align}
\label{AppAKinCoef1Aux2}
L^{xy}_{\a\b} &=\frac{1}{h}\Eint \; \zeta^x_E\zeta^y_E \Big(\d_{\a\b} 
	-\av{\abs{\AS^{\a\b}_{E,\df}}^2}\Big)f'_E,\\[3pt]
\label{AppAKinCoef2Aux2}
L^{\w x}_\a &= \frac{T}{2h}\Eint\bsum\Big(	
	2\zeta^\w {{{\rm Im}}}\Bigl[
	\av{\dAS^{\a\b}_{E,\df}\AS^{\a\b\ast}_{E,\df}}\Bigr]\d_{xq}
	-4\zeta^x_E{{{\rm Re}}}\Bigl[\av{\AS^{\a\b\ast}_{E,\df}\A^{\a\b}_E}
	\Bigr]
	+\zeta^x_E \zeta^\w\de	{{{\rm Im}}}\Bigl[
	\av{\dAS^{\a\b}_{E,\df}\AS^{\a\b\ast}_{E,\df}}\Bigr]\Bigr)f_E\\
&= \frac{T}{h}\Eint\bsum \Bigl(
	\zeta^\w {{{\rm Im}}}\Bigl[\av{\dAS^{\a\b}_{E,\df}
	\AS^{\a\b\ast}_{E,\df}}\Bigr]\d_{xq}
	+\zeta^x_E\zeta^\w\de {{{\rm Im}}}\Bigl[
	\av{\dAS^{\a\b}_{E,\df}\AS^{\a\b\ast}_{E,\df}}\Bigr]\Bigr)f_E
	\nonumber\\
&= \frac{1}{h}\Eint\;\zeta^x_E\zeta^\w\bsum {{{\rm Im}}}\Bigr[
	\av{\dAS^{\a\b}_{E,\df}\AS^{\a\b\ast}_{E,\df}}\Bigr]f'_E,
	\nonumber\\[3pt]
\label{AppAKinCoef3Aux2}
L^{x\w}_\a &= \frac{1}{h}\Eint \; \zeta^x_E \zeta^\w \bsum
	{{{\rm Im}}}\Bigl[\av{\dAS^{\b\a\ast}_{E,\df}
	\AS^{\b\a}_{E,\df}}\Bigr],\\[3pt]
\label{AppAKinCoef4Aux2}
L^{\w\w}&=  \frac{T}{2h}\Eint \;\zeta^\w \absum \Bigl(
	4{{{\rm Im}}}\Bigl[\av{\dAS^{\a\b\ast}_{E,\df}\A^{\a\b}_E}\Bigr]
	+\zeta^\w\de\av{\abs{\dAS^{\a\b}_{E,\df}}^2}\Bigr)f_E\\
&= \frac{1}{2h}\Eint \; (\zeta^\w)^2\absum
	\av{\abs{\dAS^{\a\b}_{E,\df}}^2}f'_E. \nonumber 
\end{align}
\end{subequations}
\end{widetext}
Here, we have used the sum rule 
\begin{align}
4\bsum{{{\rm Re}}}\Bigl[\av{\AS^{\a\b\ast}_{E,\df}\A^{\a\b}_E}\Bigr]
&= -i\zeta^\w\bsum \de\av{\dAS^{\a\b\ast}_{E,\df}\AS^{\a\b}_{E,\df}}
	\\
&= -\zeta^\w\bsum\de {{{\rm Im}}}\Bigl[
	\av{\dAS^{\a\b}_{E,\df}\AS^{\a\b\ast}_{E,\df}}\Bigr], \nonumber
\end{align}
which follows from Eq.~\eqref{Lem_2b}, in the first step of 
Eq.~\eqref{AppAKinCoef2Aux2};
in the second step, we have applied an integration by parts with
respect to $E$ noting that
$\partial_E \zeta^x_E = \d_{xq}$ and $\partial_E f_E = - f'_E/T$. 
In Eq.~\eqref{AppAKinCoef4Aux2}, we have used that 
\begin{equation}
\absum {{{\rm Im}}}\Big[\av{\dAS^{\a\b\ast}_{E,\df}\A^{\a\b}_E}\Big]
	=0
\end{equation}
as a consequence of the sum rule \eqref{Lem_2c} and the fact that the
period average is a linear operation; 
finally we have integrated the remaining term by parts. 

\subsection{Amplitude Picture}\label{Appx_KinCoeff_Wd}
The amplitude picture provides an alternative way of extending the 
framework of irreversible thermodynamics to periodically driven 
coherent conductors. 
This theory, which is complementary to the one discussed in 
Sec.~\ref{SecFAR}, can be developed along the same lines as for
stochastic systems, see Refs.~\cite{Brandner2015f,Proesmans2015a,
Brandner2016,Proesmans2015}. 
We first divide the single-carrier Hamiltonian that describes the
dynamics inside the conductor into two contributions, 
\begin{equation}\label{App2Hamiltonian}
H_t \equiv H_0 + \Delta U_t,
\end{equation}
where free the Hamiltonian $H_0$ is time-independent, the dynamical
scattering potential $U_t=U_{t+\T}$ accounts for the periodic driving
and the parameter $\Delta$ denotes to the amplitude of this 
perturbation.
The Floquet scattering amplitudes can thus be decomposed as 
\begin{equation}\label{App2ScattExp}
S^{\a\b}_{E\ix{n},E} =\d_{n0} S^{\a\b}_E 
	+ \Delta Z^{\a\b}_{E\ix{n},E} 
\end{equation}
with the first contribution describing elastic scattering events and
the second one accounting for inelastic events, which are induced by
the time-dependent driving.
We now introduce the affinity $F^w\equiv \Delta/T$ and the work flux
\begin{equation}\label{AppAWorkFl}
J^w\equiv \Pi_{{{\rm ac}}}/\Delta
	= \frac{\Delta}{\T}\Eint \absum\nsum n\abs{Z^{\a\b}_{\En,E}}^2f^\b_E,
\end{equation}
where the second expression is obtained by inserting 
Eq.~\eqref{App2ScattExp} into Eq.~\eqref{TwAcPower}. 
With these definitions, the total rate of entropy production 
\eqref{TwEntProdAux} assumes the canonical bilinear form 
\begin{equation}
\sigma=J^wF^w +\asum\xsum J^x_\a F^x_\a. 
\end{equation}

The frequency and the amplitude picture are equivalent on the general 
level. 
Their corresponding linear-response theories, however, describe two
physically different regimes, where either the speed or the strength
of the periodic driving is small. 
This difference becomes particularly clear from the kinetic 
coefficients in the amplitude picture, which are given by 
\footnote{To derive these expressions, insert the first-order 
expansions \eqref{TwFermiFunctExp} and 
$S^{\a\b}_{\En,E}=\d_{n0}S^{\a\b}_E +\Delta\mathcal{Z}^{\a\b}_{En,E}$
into the Eqs. \eqref{AppAKinCurr} and \eqref{AppAWorkFl} and use the
sum rule 
\begin{equation*}
\bsum {{{\rm Re}}}\Bigl[ S^{\a\b}_E \mathcal{Z}^{\a\b\ast}_{E,E}\Bigr]
	=0,
\end{equation*}
which follows from the unitarity condition \eqref{TwSumRulesSMr}.
}
\begin{align}
G^{xy}_{\a\b}&\equiv \partial_{F^y_\b}J^x_\a|_{{{\rm eq}}}
	= \frac{1}{h}\Eint \; \zeta^x_E \zeta^y_E \Bigl(
	\d_{\a\b} - \abs{S^{\a\b}_E}^2\Bigr) f'_E,\\[3pt]
G^{wx}_\a &\equiv \partial_{F^w}J^x_\a |_{{{\rm eq}}}
	= 0,\\[3pt]
G^{xw}_\a &\equiv \partial_{F^x_\a} J^w|_{{{\rm eq}}}
	=0, \\[3pt]
G^{ww} &\equiv \partial_{F^w}J^w|_{{{\rm eq}}}
	= \frac{T}{\T}\Eint \absum n \abs{\mathcal{Z}^{\a\b}_{\En,E}}^2
	f_E
\end{align}
with $\mathcal{Z}^{\a\b}_{\En,E}\equiv Z^{\a\b}_{\En,E}|_{\Delta=0}$.
Hence, the thermo-chemical variables, $F^x_\a$ and $J^x_\a$ decouple
from the mechanical ones, $F^w$ and $J^w$, in the weak-driving regime.
By contrast, this coupling persists in slow-driving regime, provided 
that the driving protocols are not symmetric under time reversal, cf.
Sec.~\ref{Sec_OCRel}. 

\section{Some useful Lemmas}\label{App_Lemmas}
In this appendix, we collect a series of sum rules for the Floquet 
scattering amplitudes in the slow-driving regime along with 
sketches of their proofs.  
Further details on the derivations may be found in 
Refs.~\cite{Moskalets2012,Ludovico2015b}. 

\begin{lemma}\label{Lem_1}
Let $S^{\a\b}_{\En,E}$ be the Floquet scattering amplitudes of a 
multi-terminal conductor that is subject to the periodic driving
fields $\df$ with frequency $\w\equiv 2\pi/\T$, then
\begin{equation}
\nsum \abs{S^{\a\b}_{\En,E}}^2 = \mathcal{X}_0 
	+ \frac{\hbar\w}{2}\mathcal{X}_1, \quad
\nsum n\abs{S^{\a\b}_{\En,E}}^2 = \mathcal{Y}_0 
	+ \frac{\hbar\w}{2}\mathcal{Y}_1
\end{equation}
up to second-order corrections in $\hbar\w$ with
\begin{subequations}\label{Lem_1all}
\begin{align}
\label{Lem_1a}
&\mathcal{X}_0= 
	\av{\abs{\AS^{\a\b}_{E,\df}}^2},\\[3pt]
\label{Lem_1b}
&\mathcal{X}_1=
	4{{{\rm Re}}}\Bigr[\bigl\llangle\AS^{\a\b\ast}_{E,\df}\A^{\a\b}_E
		\bigr\rrangle\Bigr]
	-\frac{1}{\w}\de{{{\rm Im}}}
 		\Bigl[\bigl\llangle	\dot{\AS}^{\a\b}_{E,\df}\AS^{\a\b\ast}_{E,\df}
 		\bigr\rrangle\Bigr],\\[3pt]
\label{Lem_1c}
&\mathcal{Y}_0=
	- \frac{1}{\w}{{{\rm Im}}}\Bigr[\bigl\llangle\dot{\AS}^{\a\b}_{E,\df}
	\AS^{\a\b\ast}_{E,\df}\bigr\rrangle\Bigr],\\[3pt]
\label{Lem_1d}
&\mathcal{Y}_1=
	\frac{4}{\w}{{{\rm Im}}}\Bigl[\bigl\llangle
			\dot{\AS}^{\a\b\ast}_{E,\df}\A^{\a\b}_E\bigr\rrangle\Bigr]
	+\frac{1}{\w^2}\partial_E^{\phantom{\b}}\bigl\llangle\abs{
			\dot{\AS}^{\a\b}_{E,\df}}^2\bigr\rrangle.
\end{align}
Here, $\AS^{\a\b}_{E,\df}$ denotes the frozen scattering amplitudes, 
$\A^{\a\b}_E$ the non-adiabatic corrections defined in 
Eq.~\eqref{TwAdApproxFSA} and $\llangle \cdots\rrangle$ indicates the 
time average over one period $\T$. 
\end{subequations}
\end{lemma}
\begin{proof}
Insert the low-frequency expansion \eqref{TwAdApproxFSA} for 
$S^{\a\b}_{\En,E}$ and collect all zeroth and first order terms in
$\hbar\w$. 
Perform the sum over all integers $n$ by using the symbolic identity 
\begin{equation}\label{Lem_1Aux}
\nsum n^k e^{i n\w(t-t')}=
	\frac{\T}{(i\w)^k}\partial^k_t\d_{t-t'},
\end{equation}
where $k=0,1,\dots$ and $t,t'\in [0,\T)$. 
Integrate by parts with respect to $t$ as needed. 
\end{proof}

\begin{lemma}\label{Lem_2}
The frozen scattering amplitudes and the non-adiabatic corrections 
obey the joint sum rules 
\begin{subequations}
\begin{align}
\label{Lem_2a}
&\begin{aligned}[t]
	& 2i\w\asum\bigl(\AS^{\a\b}_{E,\dft}\A^{\a\g\ast}_{E,t}
		+\AS^{\a\g\ast}_{E,t}\A^{\a\b}_{E,\dft}\bigr)\\
	&\hspace*{1cm}=\asum\bigl(
		\dAS^{\a\b}_{E,\dft}\de\AS^{\a\g\ast}_{E,\dft}
		-\dAS^{\a\g\ast}_{E,\dft}\de\AS^{\a\b}_{E,\dft}\bigr),
\end{aligned}\\[3pt]
\label{Lem_2b}
&\begin{aligned}[t]
	& 2i\w\asum\bigl(\AS^{\b\a}_{E,\dft}\A^{\g\a\ast}_{E,t}
		+\AS^{\g\a\ast}_{E,\dft}\A^{\b\a}_{E,t}\bigr)\\
	&\hspace*{1cm} =\asum\bigl(
		\dAS^{\g\a\ast}_{E,\df\ix{t}}\de\AS^{\b\a}_{E,\dft}
		-\dAS^{\b\a}_{E,\dft}\de\AS^{\g\a\ast}_{E,\dft}\bigr),
\end{aligned}\\[3pt]
\label{Lem_2c}
&\absum{{{\rm Im}}}\Bigl[\dAS^{\a\b\ast}_{E,\dft}\A^{\a\b}_{E,t}\Bigr]
	=0
\end{align}
\end{subequations}
for any $t\in [0,\T)$. 
\end{lemma}
\begin{proof}
We begin with the sum rule \eqref{Lem_2a}. 
Inserting the expansion \eqref{TwAdApproxFSA} into the unitarity 
condition \eqref{TwSumRulesSMl}, collecting first-order terms in
$\hbar\omega$ and carrying out the sum over $n$ using 
Eq.~\eqref{Lem_1Aux} shows that 
\begin{align}
\label{Lem_2Aux1}
\tint\Bigl\{\asum\Bigl(
	& 2\dAS^{\a\b}_{E,\dft}\de\AS^{\a\g\ast}_{E,\dft}\\
	& + \AS^{\a\b}_{E,\dft}\de\dAS^{\a\g\ast}_{E,\dft}
		+ \AS^{\a\g\ast}_{E,\dft}\de\dAS^{\a\b}_{E,\dft}\nonumber\\
	& -2i\w\bigl(\AS^{\a\b}_{E,\dft}\A^{\a\g\ast}_{E,t}
		+\AS^{\a\g\ast}_{E,\dft}\A^{\a\b}_{E,t}\bigr)\Bigr)\Bigr\}
		e^{im\w t}=0,\nonumber
\end{align}
for every integer $m$ and any indices $\b$ and $\g$. 
This condition can only be met if the expression inside the curly 
brackets vanishes for every $t\in [0,\T)$, that is, if 
\begin{equation}\label{Lem_2Aux2}
2i\w(\Ab^\dagger\Sb + \Ab\Sb^\dagger)
	=2\Sb'^\dagger\dot{\Sb}	+\dot{\Sb}'^\dagger\Sb
	+\Sb^\dagger \dot{\Sb}',
\end{equation}
where we introduced the matrices $\Sb$ and $\Ab$ with elements 
$(\Sb)_{\a\b}\equiv \AS^{\a\b}_{E,\dft}$ and $(\Ab)_{\a\b}\equiv
\A^{\a\b}_{E,t}$ to simplify the notation and primes indicate 
derivatives with respect to $E$. 
Next, we note that the matrix $\Sb$ is unitary, i.e.,
$\Sb^\dagger\Sb=\Sb\Sb^\dagger=\mathbbm{1}$, owing to the 
conditions \eqref{TwARSumRulesSM}. 
Therefore, we have $\partial_t\partial_E\Sb^\dagger\Sb=0$ and writing 
out the derivatives gives 
\begin{equation}
\dot{\Sb}'^\dagger\Sb+\Sb^\dagger\dot{\Sb}'\\
	= -\Sb'^\dagger\dot{\Sb}-\dot{\Sb}^\dagger\Sb'. 
\end{equation}
Inserting this relation into Eq.~\eqref{Lem_2Aux2} yields the result
\begin{equation}\label{Lem_2Aux3}
2i\w(\Ab^\dagger \Sb+\Ab\Sb^\dagger)
	=\Sb'^\dagger\dot{\Sb}-\dot{\Sb}^\dagger\Sb',
\end{equation}
which is equivalent to the first sum rule \eqref{Lem_2a}.
The second sum rule \eqref{Lem_2b} can be derived along the same lines
starting with the unitarity condition \eqref{TwSumRulesSMr}. 

To derive the third sum rule \eqref{Lem_2c}, we first note that 
\begin{align}
2i\absum {{{\rm Im}}}\Bigl[\dAS^{\a\b\ast}_{E,\dft}\A^{\a\b}_{E,t}
	\Bigr]
	& = \tr{\dot{\Sb}^\dagger\Ab-\Ab^\dagger\dot{\Sb}}\\
	& = \tr{\Ab\dot{\Sb}^\dagger\Sb\Sb^\dagger
		-\Ab^\dagger\dot{\Sb}\Sb^\dagger\Sb}
	\nonumber\\
	& = \tr{\bigl(\Sb^\dagger\Ab+\Ab^\dagger\Sb\bigr)
		\dot{\Sb}^\dagger\Sb},
	\nonumber
\end{align}
where we have used that $\dot{\Sb}\Sb^\dagger=-\Sb\dot{\Sb}^\dagger$,
since $\Sb$ is unitary. 
Second, eliminating the matrix $\Ab$ with the help of the relation 
\eqref{Lem_2Aux3} gives 
\begin{align}
-4\w\absum {{{\rm Im}}}\Bigl[\dAS^{\a\b\ast}_{E,\dft}\A^{\a\b}_{E,t}
	\Bigr]
	&= \tr{\bigl(\Sb'^\dagger\dot{\Sb}-\dot{\Sb}^\dagger\Sb'\bigr)
		\dot{\Sb}^\dagger\Sb}\\
	&= \tr{\bigl(\Sb^\dagger\Sb'\dot{\Sb}^\dagger\Sb
		-\dot{\Sb}^\dagger\Sb'\bigr)\dot{\Sb}^\dagger\Sb}
		\nonumber\\
	&= 0,\nonumber
\end{align}
where the second line follows by noting that $\Sb\Sb'^\dagger=
-\Sb'\Sb^\dagger$ and $\dot{\Sb}\Sb^\dagger=-\Sb\dot{\Sb}^\dagger$.
\end{proof}

\begin{lemma}\label{Lem_3}
The frozen scattering amplitudes can be chosen such that they obey the
sum rule 
\begin{equation}\label{Lem_3a}
\absum \bigl\llangle\dAS^{\a\b}_{E,\df}\AS^{\a\b\ast}_{E,\df}
	\bigr\rrangle=0.
\end{equation}
\end{lemma}
\begin{proof}
The frozen scattering matrix admits the spectral decomposition 
$\Sb=\sum_k |\phi_k\rangle\langle\phi_k| e^{i\phi_k}$, where the 
$\phi_k$ are real and the $|\phi_k\rangle$ are normalized orthogonal
eigenvectors. 
Since $\Sb$ is a $\T$-periodic function of time, the phases 
$\phi_k$ further have to obey the condition
$\llangle\dot{\phi}_k\rrangle=\w M_k$ with $M_k$ being an integer.   
It follows that
\begin{align}
\absum\bigl\llangle\dot{\AS}^{\a\b}_{E,\df}\AS^{\a\b\ast}_{E,\df}
	\bigr\rrangle 
	& = \av{\tr{\dot{\Sb}\Sb^\dagger}}\\
	& = \sum\nolimits_k \av{i\dot{\phi}_k 
		+ \partial_t \langle\phi_k|\phi_k\rangle}\nonumber\\
	& = i\w \sum\nolimits_k M_k\equiv i \w M,\nonumber
\end{align}
Thus, the sum rules \eqref{Lem_3a} can be enforced by applying the 
gauge transformation $\AS^{\a\b}_{E,\df}\rightarrow 
\AS^{\a\b}_{E,\df}e^{-iM\w t}$ to the frozen scattering amplitudes.
\end{proof}

\begin{lemma}\label{Lem4}
Let $S^{\a\b}_{\En,E}$ be a set of Floquet scattering amplitudes that 
obey the unitarity conditions \eqref{TwSumRulesSM} and $f^\a_E$ the 
Fermi function \eqref{TwFermiF}, then
\begin{align}
\sigma &\equiv \frac{1}{h}\Eint\absum\nsum 
	\left(\frac{E_n-\mu_\a}{T_\a}-\frac{E-\mu_\b}{T_\b}\right)
	\abs{S^{\a\b}_{\En,E}}^2 f^\b_E\nonumber\\
\label{Lem4_0}
&\geq \frac{2}{h}\Eint \absum\nsum \abs{S^{\a\b}_{\En,E}}^2
	(f^\a_{\En}-f^\b_E)^2.
\end{align}
\end{lemma}
\begin{proof}
We first observe that, using the conditions \eqref{TwSumRulesSM}, 
$\sigma$ can be expressed as \cite{Brandner2020}
\begin{equation}\label{Lem4_1}
\sigma = \frac{1}{h}\Eint\absum\nsum \abs{S^{\a\b}_{\En,E}}^2
	H^{\a\b}_{\En,E}
\end{equation}
with 
\begin{equation}
H^{\a\b}_{\En,E} \equiv \eta[f^\a_{\En}]-\eta[f^\b_E]
	-\eta'[f^\a_{\En}]\bigl(f^\a_{\En}-f^\b_E\bigr),
\end{equation}
$\eta[x]\equiv -x\ln[x] -(1-x)\ln[1-x]$ and primes indicating
derivatives. 
Next, we note that, by Taylor's theorem, there exists a $g$ between
$f^\a_{\En}$ and $f^\b_E$ such that 
\begin{equation}
H^{\a\b}_{\En,E} = - \frac{\eta''[g]}{2}(f^\a_{\En}-f^\b_E)^2.
\end{equation}
Since the Fermi function \eqref{TwFermiF} takes only values between 
$0$ and $1$, it follows that $g\in[0,1]$. 
Therefore, we have $-\eta''[g]= 1/g+1/(1-g)\geq 4$ and 
$H^{\a\b}_{\En,E}\geq 2 (f^\a_{\En}-f^\b_E)^2$. 
Inserting this bound into Eq.~\eqref{Lem4_1} completes the proof. 
\end{proof}

\begin{lemma}\label{Lem5}
Let $S^{\a\b}_{\En,E}$ be a set of Floquet scattering amplitudes that 
obey the unitarity conditions \eqref{TwSumRulesSM} and  
$f^\a_E$ the Fermi function \eqref{TwFermiF}. 
Then, for $\hat{S}^{\a\b}_{\En,E}\equiv 
(1-\d_{n0}\d_{\a\b})S^{\a\b}_{\En,E}$ and $f'^\a_E\equiv 
f^\a_E(1-f^\a_E)$, we have 
\begin{equation}\label{Lem5_1}
\Psi^{\r\r}_\a \equiv \frac{1}{h}\Eint \bsum\nsum 
	\abs{\hat{S}^{\a\b}_{\En,E}}^2(f^\a_{\En}-f^\b_E)^2
	\leq \frac{\sigma}{2}
\end{equation}
and
\begin{equation}\label{Lem5_2}
\Omega^{\r\r}_\a \equiv \frac{2}{h}\Eint\nsum 
	\abs{\hat{S}^{\a\a}_{\En,E}}^2f'^\a_E
	\leq \frac{2T_\a}{h},
\end{equation}
where $\sigma$ is defined in Eq.~\eqref{Lem4_0}. 
\end{lemma}
\begin{proof}
For Eq.~\eqref{Lem5_1}, observe that 
\begin{align}
2\Psi^{\r\r}_\a &\leq \frac{2}{h}\Eint\absum\nsum
	\abs{\hat{S}^{\a\b}_{\En,E}}^2
	(f^\a_{\En}-f^\b_E)^2\\
&= \frac{2}{h}\Eint\absum\nsum\abs{S^{\a\b}_{\En,E}}^2
	(f^\a_{\En}-f^\b_E)^2\nonumber
\end{align}
and use Lemma~\ref{Lem4}.
For Eq.~\eqref{Lem5_2}, note that $f'^\a_E\geq 0$ and, by the 
unitarity conditions \eqref{TwSumRulesSM}, $\sum_n
\abs{\hat{S}^{\a\a}_{\En,E}}^2\leq\bsum\nsum \abs{S^{\b\a}_{\En,E}}^2=
1$ such that 
\begin{equation}
\Omega^{\r\r}_\a \leq \frac{2}{h}\Eint \; f'_E = \frac{2T_\a}{h}
	\frac{\varphi_\a}{1+\varphi_\a}\leq \frac{2T_\a}{h},
\end{equation}
where $\varphi_\a\equiv \exp[\mu_\a/T_\a]$. 
\end{proof}

\section{Thermodynamic Uncertainty Relation in Linear Response}\label{Apx_TURLR}

In this appendix, we show that the thermodynamic uncertainty relation
\begin{equation}\label{AppC_TUR}
\frac{\sigma P^{xx}_{\a\a}}{(J^x_\a -\w\partial_\w J^x_\a)^2}
	\geq 2,
\end{equation}
which was derived in Ref.~\cite{Koyuk2019b} for periodically driven 
Markov jump processes, holds for coherent transport in linear 
response if the frozen scattering amplitudes obey the symmetry 
\begin{equation}\label{AppC_Sym}
\AS^{\a\b}_{E,\df}=\AS^{\b\a}_{E,\df}. 
\end{equation} 

We first recall Sec.~\ref{Sec_KinCoeff} and note that the symmetry 
\eqref{AppC_Sym} implies $L^{x\w}_\a=-L^{\w x}_\a$. 
As a result, the total rate of entropy production 
\eqref{TwSecondLawCan} can be divided into two contributions, 
$\sigma= \sigma_{{{\rm th}}}+\sigma_{{{\rm ac}}}$
with 
\begin{subequations}
\begin{align}
\label{AppC_EntDec1}
\sigma_{{{\rm th}}} &\equiv \absum\xysum L^{xy}_{\a\b} F^x_\a F^y_\b\\
&= \frac{1}{2h}\Eint \absum
	\bigl\llangle \abs{\hat{\AS}_{E,\df}^{\a\b}}^2\bigr\rrangle
	\bigl(Y^{\a\b}_E\bigr)^2 f'_E
    \geq 0\quad\text{and}\nonumber\\[3pt]
\label{AppC_EntDec2}
\sigma_{{{\rm ac}}} &\equiv L^{\w\w} (F^\w)^2 \geq 0,
\end{align}
\end{subequations}
where $\hat{\AS}^{\a\b}_{E,\df}\equiv(1-\d_{\a\b})\AS^{\a\b}_{E,\df}$, 
$Y^{\a\b}_E\equiv \sum_x \zeta^x_E (F^x_\a-F^x_\b)$ 
and the second line in Eq.~\eqref{AppC_EntDec1} follows by using the
unitarity conditions \eqref{TwARSumRulesSM}.
We now consider the quadratic form 
\begin{align}\label{AppC_QuadF}
\Xi^x_\a &\equiv \sigma_{{{\rm th}}}+2G\bsum\ysum L^{xy}_{\a\b}F^y_\b 
	+ G^2L^{xx}_{\a\a}\\
&= \frac{1}{2h}\Eint \sum\nolimits_{\b\neq\a}\Bigl\{
	\begin{aligned}[t] &\sum\nolimits_{\g\neq\a}
		\bigl\llangle \abs{\hat{\AS}^{\b\g}_{E,\df}}^2\bigr\rrangle
		\bigl(Y^{\a\b}_E\bigr)^2\\
	&+2\bigl\llangle \abs{\hat{\AS}^{\a\b}_{E,\df}}^2\bigr\rrangle
	\bigl(\zeta^x_E G +Y^{\a\b}_E\bigr)^2\Bigr\} f'_E, 
	\end{aligned}\nonumber
\end{align}  
where $G$ is real and otherwise arbitrary. 
The second line in Eq.~\eqref{AppC_QuadF}, which follows from 
the unitarity conditions \eqref{TwARSumRulesSM} and the symmetry 
\eqref{AppC_Sym}, proves that $\Xi^x_\a$ is positive semi-definite. 
Next, recalling the kinetic equations \eqref{TwARKinEq} and the 
fluctuation-dissipation theorem \eqref{TwFDT}, yields 
\begin{equation}
\bsum\ysum L^{xy}_{\a\b} F^y_\b = J^x_\a - \w\partial_\w J^x_\a,
	\quad 	L^{xx}_{\a\a} = P^{xx}_{\a\a}\eq/2. 
\end{equation}
Inserting these identities into the definition of $\Xi^x_\a$ and 
noting that $\sigma_{{{\rm th}}}\leq\sigma$ since $\sigma_{{{\rm ac}}}
\geq 0$ shows that 
\begin{equation}
\sigma +2G(J^x_\a - \w\partial_\w J^x_\a) 
	+G^2 P^{xx}_{\a\a}\eq/2\geq \Xi^x_\a \geq 0 
\end{equation}
for any $G$.
Minimizing the left-hand side of this inequality with respect to $G$ 
finally gives the thermodynamic uncertainty relation \eqref{AppC_TUR}.

\section{Quantum Generator}\label{Apx_QGen}
\begin{figure}
\includegraphics[scale=1.05]{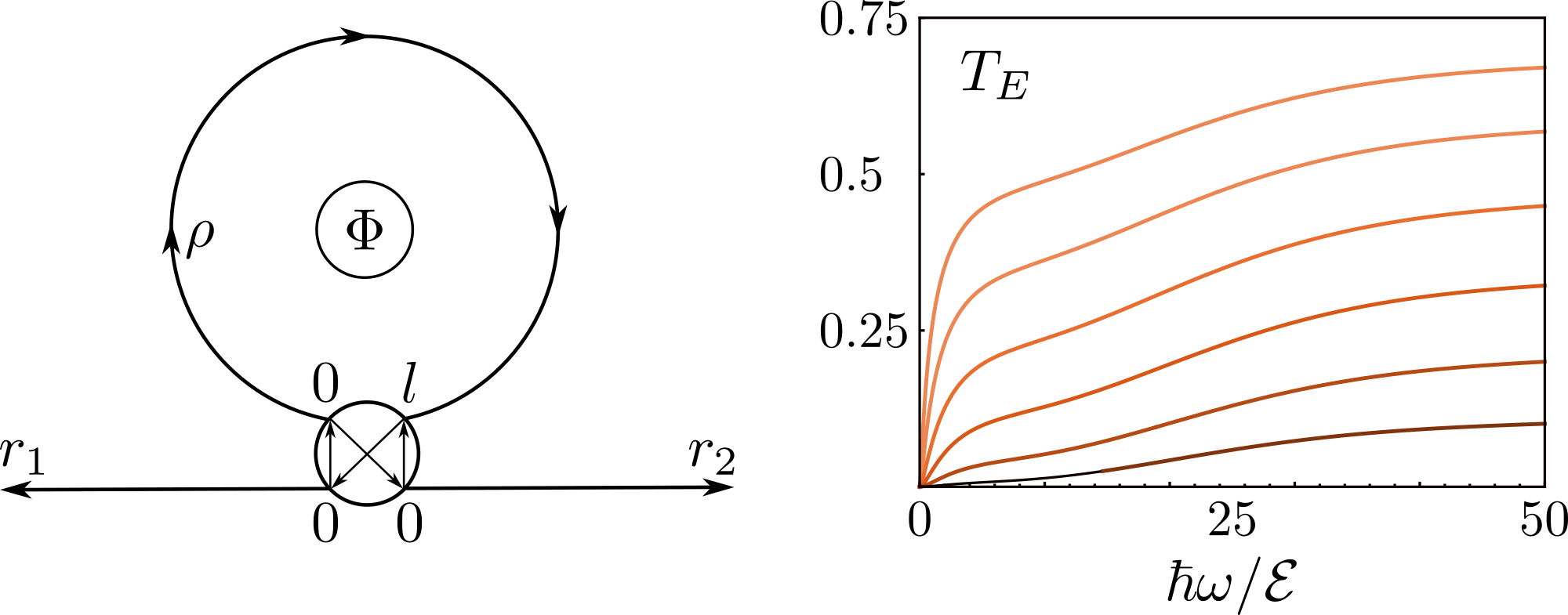}
\caption{Scattering Theory of a Quantum Generator.
The system consists of two leads, an ideal beam splitter and a loop 
with length $l$, which is subject to the magnetic flux $\Phi$.
The leads and the loop are parameterized in terms of the coordinates
$r_1,r_2\in [0,\infty)$ and $\rho\in [0,l]$. 
Incoming waves on the leads $1$ and $2$ enter the loop at the 
positions $\r=0$ and $\r=l$, respectively;
clockwise and counter-clockwise propagating waves on the loop exit at 
the positions $\r=l$ and $\r=0$ and become outgoing waves on the leads 
$1$ and $2$, respectively. 
The plot shows how the transmission function \eqref{AppC_TFun} depends
on the driving frequency $\w$, i.e, the rate at which the magnetic
flux $\Phi$ changes, for $E/\mathcal{E}=2^0/10, \dots, 2^5/10$ 
from top to bottom. 
\label{Fig_AQG}}
\end{figure}

In this appendix, we derive the exact Floquet scattering amplitudes 
for the quantum generator discussed in Secs.~\ref{SecAdChrPumpTM} and
\ref{Sec_FFEQGR} by solving the corresponding Floquet-Schr\"odinger
equation in position representation \cite{Brandner2020}. 
We further provide explicit formulas for the quantities entering 
the relations \eqref{FFE_ER} and \eqref{FFEQJ_BFig}.

\subsection{Scattering Amplitudes}
The system is parameterized according to Fig.~\ref{Fig_AQG}. 
An incoming carrier with energy $E>0$ in the terminal $\a$ is 
described by the scattering state $|\phi^\a_{E,t}\rangle$. 
On the leads, the wave function of this state is given by 
\begin{subequations}\label{AppC_LAn}
\begin{align}
\phi^\a_{E,t}[r_1]&= \d_{\a 1}w^-_E[r_1]
	+\nsum C^{1\a}_{\En,E}w^+_{\En}[r_1] e^{-in\w t},\\[3pt]
\phi^\a_{E,t}[r_2]&= \d_{\a 2}w^-_E[r_2]
	+\nsum C^{2\a}_{\En,E}w^+_{\En}[r_2] e^{-i n\w t}, 
\end{align}
\end{subequations}
where the summations run over all integers, the $C^{\a\b}_{\En,E}$ are
yet undetermined complex coefficients and
\begin{equation}
w^\pm_E[r]\equiv  w_E \exp[\pm i k_E r]
\end{equation}
denotes a normalized incoming $(-)$ or outgoing $(+)$ plane wave with
$w_E\equiv \sqrt{M/2\pi k_E \hbar^2}$ and 
$k_E\equiv\sqrt{2ME/\hbar^2}$ \cite{Brandner2020}; 
we recall that $M$ denotes the carrier mass. 
On the loop, the scattering wave function obeys the 
Floquet-Schr\"odinger equation 
\begin{equation}\label{AppC_FSE}
E\phi^\a_{E,t}[\rho]= 
	-\frac{\hbar^2}{2M}\bigl(\partial_\rho -i\w t/l\bigr)^2
	\phi^\a_{E,t}[\rho]
	-i\hbar\partial_t\phi^\a_{E,t}[\rho], 
\end{equation}
with respect to the boundary conditions 
\begin{subequations}\label{AppC_BCP}
\begin{align}
&\phi^1_{E,t}[\r]\big|_{\r=0} = w_E
	+\nsum C^{21}_{\En,E}w_{\En} e^{-i n\w t}\\[3pt]
&\phi^1_{E,t}[\r]\big|_{\r=l} =
	\nsum C^{11}_{\En,E} w_{\En}  e^{-i n\w t},\\[3pt]
&\hat{P}_t\phi^1_{E,t}[\r]\big|_{\r=0}
	= w'_E-\nsum C^{21}_{\En,E} w'_{\En} e^{-in\w t},
	\label{AppC_BCPc}\\[3pt]
&\hat{P}_t\phi^1_{E,t}[\r]\big|_{\r=l}
	= \nsum C^{11}_{\En,E} w'_{\En} e^{-in\w t}
\end{align}
\end{subequations}
and 
\begin{subequations}\label{AppC_BCM}
\begin{align}
&\phi^2_{E,t}[\r]\big|_{\r=0} = 
	\nsum C^{22}_{\En,E} w_{\En} e^{-in\w t},\\[3pt]
&\phi^2_{E,t}[\r]\big|_{\r=l} = w_E
	+\nsum C^{12}_{\En,E} w_{\En} e^{-in\w t}\\[3pt]
&\hat{P}_t\phi^2_{E,t}[\r]\big|_{\r=0}	
	= -\nsum C^{22}_{\En,E} w'_{\En} e^{-in\w t},
	\label{AppC_BCMc}\\[3pt]
&\hat{P}_t\phi^2_{E,t}[\r]\big|_{\r=l}
	= -w'_E + \nsum C^{12}_{\En,E} w'_{\En} e^{-in\w t},
\end{align}
\end{subequations}
where $w'_E\equiv  w_E k_E$ and $\hat{P}_t
\equiv -i\partial_\rho -\w t/l$.
Here, we assume that the magnetic flux $\Phi$ increases linearly in
time, that is, $\Phi_t = \hbar\omega c t/e$, where $c$ denotes the 
speed of light and $e$ the carrier charge. 
The parameter $E$ corresponds to the Floquet energy on the loop.
The boundary conditions \eqref{AppC_BCP} and \eqref{AppC_BCM}
are determined by the beam splitter that connects the loop with the
leads
\footnote{To obtain boundary conditions \eqref{AppC_BCP} and 
\eqref{AppC_BCM}, we require that both $\phi^\a_{E,t}$ and 
$p\phi^\a_{E,t}$ are continuous at the beam splitter, where the second
condition ensures the continuity of probability currents 
\cite{Schiff1968}. 
We recall that the momentum operator has the position representation 
$p\hateq -i\hbar \partial_{r_\a}$ on the leads and 
$p\hateq -i\hbar \partial_\r -\hbar\w t/l=\hat{P}_t/\hbar$ on the 
loop.
The additional factor $(-1)$ on the right-hand sides of the 
Eqs.~\eqref{AppC_BCPc} and \eqref{AppC_BCMc} accounts for the fact 
that the parameterizations of the leads and the loop run in opposite
directions at the connections $\r=0\rightarrow r_1=0$ and 
$\r=0\rightarrow r_2=0$.}. 

A general solution of Eq.~\eqref{AppC_FSE} that is compatible with 
the boundary conditions \eqref{AppC_BCP} and \eqref{AppC_BCM} is given
by 
\begin{equation}\label{AppC_SolG}
\phi^\a_{E,t}[\rho]= \nsum \bigl(a^\a_n {{{\rm Ai}}}_{\En}[\rho]
	+b^\a_n {{{\rm Bi}}}_{\En}[\rho]\bigr)e^{i(\rho/l-n)\w t} ,
\end{equation}
where $a^\a_n$ and $b^\a_n$ are arbitrary complex coefficients.
The modified Airy functions are thereby defined as 
\begin{equation}
{{{\rm Xi}}}_{E}[\rho] \equiv
	{{{\rm Xi}}}\bigl[(\hbar\w/\mathcal{E})^\frac{1}{3}(\rho/l-E/\hbar\w)
	\bigr]
\end{equation}
in terms of the standard Airy functions ${{{\rm Xi}}}\equiv{{{\rm Ai}}}, 
{{{\rm Bi}}}$ \cite{Olver1974}. 
The parameter $\mathcal{E}\equiv \hbar^2/2M l^2$ sets the natural energy 
scale of the system. 
Inserting the solution \eqref{AppC_SolG} into the boundary conditions 
\eqref{AppC_BCP} and \eqref{AppC_BCM} and collecting Fourier
components yields two sets of linear equations, which can be written
compactly as 
\begin{subequations}\label{AppC_BCExpl1}
\begin{align}
&\mathbb{A}_{\En}\mathbf{a}_n^1 = \d_{n0}\mathbf{1} + C^{21}_{\En,E}
	\hat{\mathbf{1}},\\[3pt]
&\mathbb{A}_{E\ix{n-1}}\mathbf{a}_n^1
	=C^{11}_{E\ix{n-1},E}\mathbf{1}
\end{align}
\end{subequations}
and
\begin{subequations}\label{AppC_BCExpl2}
\begin{align}
& \mathbb{A}_{\En}\mathbf{a}^2_n
	=C^{22}_{\En,E}\hat{\mathbf{1}},\\[3pt]
& \mathbb{A}_{E\ix{n-1}}\mathbf{a}^2_n 
	= \delta_{n1}\hat{\mathbf{1}} + C^{12}_{E\ix{n-1},E}\mathbf{1}.
\end{align}
\end{subequations}
Here, we have used the relation ${{{\rm Xi}}}_{\En}[l]
={{{\rm Xi}}}_{E\ix{n-1}}[0]$.
Furthermore, we have introduced the vectors
$\mathbf{a}^\a_n\equiv (a^\a_n,b^\a_n)^{{{\rm t}}}$, $\mathbf{1}
\equiv (1,1)^{{{\rm t}}}$, $\hat{\mathbf{1}}\equiv (1,-1)^{{{\rm t}}}$
and the matrix 
\begin{equation}
\mathbb{A}_E\equiv \frac{1}{w_E}\left(\!\begin{array}{cc}
{{{\rm Ai}}}_E[0] & {{{\rm Bi}}}_E[0]\\[3pt]
-i{{{\rm Ai}}}'_E[0] & -i{{{\rm Bi}}}'_E[0]
\end{array}\!\right),
\end{equation}
where
\begin{align}
{{{\rm Xi}}}'_E[\r]& \equiv\partial_\r {{{\rm Xi}}}_E[\r]/k_E\\
&=(\hbar\w/\mathcal{E})^{\frac{1}{3}}(\mathcal{E}/E)^{\frac{1}{2}}
	{{{\rm Xi}}}'\bigl[
	(\hbar\w/\mathcal{E})^{\frac{1}{3}}(\r/l-E/\hbar\w)\bigr]
	\nonumber
\end{align} 
with ${{{\rm Xi}}}'\equiv {{{\rm Ai}}}',{{{\rm Bi}}}'$ denoting the
derivatives of the standard Airy functions \cite{Olver1974}. 

Solving the linear systems \eqref{AppC_BCExpl1} and 
\eqref{AppC_BCExpl2} yields
\begin{subequations}
\begin{align}
& C^{11}_{\En,E} = \d_{n(-1)}\frac{2}{\mathbf{1}^{{{{\rm t}}}}
	\mathbb{A}_E\mathbb{A}^{-1}_{E\ix{-1}}\mathbf{1}},\\[3pt]
& C^{21}_{\En,E} = \d_{n0} \frac{\hat{\mathbf{1}}^{{{{\rm t}}}}
	\mathbb{A}_E\mathbb{A}^{-1}_{E\ix{-1}}
	\mathbf{1}}{\mathbf{1}^{{{{\rm t}}}}
	\mathbb{A}_E\mathbb{A}^{-1}_{E\ix{-1}}\mathbf{1}},\\[3pt]
& C^{22}_{\En,E} = \d_{n1}\frac{2}{\hat{\mathbf{1}}^{{{{\rm t}}}}
	\mathbb{A}_E\mathbb{A}^{-1}_{E\ix{1}}\hat{\mathbf{1}}},\\[3pt]
& C^{12}_{\En,E} = \d_{n0}\frac{\mathbf{1}^{{{{\rm t}}}}
	\mathbb{A}_E\mathbb{A}^{-1}_{E\ix{1}}\hat{\mathbf{1}}}{
	\hat{\mathbf{1}}^{{{{\rm t}}}}\mathbb{A}_E
	\mathbb{A}^{-1}_{E\ix{1}}\hat{\mathbf{1}}}. 
\end{align}
\end{subequations}
The Floquet scattering amplitudes can now be determined by inserting 
these expressions into the ansatz \eqref{AppC_LAn} and comparing the 
result with the asymptotic boundary conditions for the Floquet 
scattering wave functions \cite{Brandner2020}, which are given by
\begin{subequations}
\begin{align}
&\phi^\a_{E,t}[r_1]\big|_{r_1\rightarrow\infty}= \d_{\a 1}w^-_E[r_1]
	+\nsum S^{1\a}_{\En,E}w^+_{\En}[r_1] e^{-in\w t},\\[3pt]
&\phi^\a_{E,t}[r_2]\big|_{r_2\rightarrow\infty}= \d_{\a 2}w^-_E[r_2]
	+\nsum S^{2\a}_{\En,E}w^+_{\En}[r_2] e^{-i n\w t}. 
\end{align}
\end{subequations}
Upon observing that the outgoing plane waves become decaying 
exponentials for negative energies, that is, $w^+_E[r]\propto 
\exp[-k_{|E|}r]$ for $E\leq 0$, we thus arrive at 
\begin{subequations}\label{AppC_FSA}
\begin{align}
& S^{11}_{\En,E} =0, && S^{21}_{\En,E} = C^{21}_{\En,E}
	&& (E\leq \hbar\w);\\
& S^{11}_{\En,E} = C^{11}_{\En,E}, 
	&& 	S^{21}_{\En,E} = C^{21}_{\En,E} 
	&& (E>\hbar\w);\\
& S^{22}_{\En,E} = C^{22}_{\En,E},
	&& S^{12}_{\En,E} = C^{12}_{\En,E}
	&& (E>0). 
\end{align}
\end{subequations}
This result leads to the compact expressions 
\begin{subequations}
\begin{align}
\abs{S^{11}_{\En,E}}^2 &= \d_{n(-1)}
	\Theta_{E\ix{-1}}R_{E\ix{-1}},\\[3pt]
\abs{S^{21}_{\En,E}}^2 &= \d_{n0}
	+\d_{n0}\Theta_{E\ix{-1}}(T_{E\ix{-1}}-1),\\[3pt]
\abs{S^{22}_{\En,E}}^2 &= \d_{n1} R_{E},\\[3pt]
\abs{S^{12}_{\En,E}}^2 &= \d_{n0} T_{E},
\end{align}
\end{subequations}
for the reflection and transmission probabilities, which were used in 
Sec.~\ref{Sec_FFEQGR}. 
Here, $\Theta$ denotes the Heaviside step function and the reflection
and transmission functions are given by 
\begin{subequations}
\begin{align}
& R_E\equiv \frac{4}{\abs{\mathbf{1}^{{{\rm t}}}
	\mathbb{A}_{E\ix{1}}\mathbb{A}^{-1}_E\mathbf{1}}^2}
	= \frac{4}{\abs{\hat{\mathbf{1}}^{{{\rm t}}}
	\mathbb{A}_E\mathbb{A}^{-1}_{E\ix{1}}\hat{\mathbf{1}}}^2},\\[3pt]
& T_E\equiv \frac{\abs{\hat{\mathbf{1}}^{{{\rm t}}}\mathbb{A}_{E\ix{1}}
	\mathbb{A}^{-1}_E\mathbf{1}}^2}{\abs{\mathbf{1}^{{{\rm t}}}
	\mathbb{A}_{E\ix{1}}\mathbb{A}^{-1}_E\mathbf{1}}^2}
	=\frac{\abs{\mathbf{1}^{{{\rm t}}}\mathbb{A}_E\mathbb{A}^{-1}_{E\ix{1}}
	\hat{\mathbf{1}}}^2}{\abs{\hat{\mathbf{1}}^{{{\rm t}}}
	\mathbb{A}_E\mathbb{A}^{-1}_{E\ix{1}}\hat{\mathbf{1}}}^2}
\label{AppC_TFun}
\end{align}
\end{subequations}
and obey the sum rule $T_E+R_E=1$. 
Using this relation, it is straightforward to verify that the Floquet 
scattering amplitudes \eqref{AppC_FSA} obey the unitarity conditions 
\eqref{TwSumRulesSM}. 
For illustration, the transmission function $T_E$ is plotted in 
Fig.~\ref{Fig_AQG} for different energies. 

\subsection{Thermodynamic Quantities}
Upon inserting the scattering amplitudes \eqref{AppC_FSA} into the 
general expressions \eqref{TwMeanCurrents}, it is now straightforward
to determine the mean values of the matter currents $J^\r_\a$ and 
the heat currents $J^q_\a = J^\ve_\a-\mu_\a J^\r_E$, which are given 
by
\begin{subequations}
\begin{align}
J^\r_1 &= - J^\r_2 = \frac{1}{h}\Eint \;\Bigl\{f^1_E-R_E f^1_{E\ix{1}}
	-T_E f^2_E\Bigr\},\\[3pt]
J^q_1  &=\frac{1}{h}\Eint\;\zeta^{q,1}_E\Bigl\{f^1_E-R_E f^1_{E\ix{1}}
	-T_E f^2_E\Bigr\},\\[3pt]
J^q_2  &= \frac{1}{h}\Eint \;\Bigl\{ \zeta^{q,2}_E(f^2_E-f^1_E) 
	+\zeta^{q,2}_{E\ix{1}}R_E (f^1_{E\ix{1}}-f^2_E)\Bigr\},
\end{align}
\end{subequations}
Analogously, the expressions \eqref{TwMENTh} yields the thermal 
fluctuations of the matter and heat currents, $D^{\r\r}_{\a\a}$ and
$D^{qq}_{\a\a}= \sum\nolimits_{uv}(\d_{u\ve}-\mu_\a\d_{u\r})(\d_{v\ve}
-\mu_\a\d_{v\r})D^{uv}_{\a\a}$,
\begin{subequations}
\begin{align}
D^{\r\r}_{11} &= D^{\r\r}_{22}
	=  \frac{1}{h}\Eint\; \Bigl\{f'^1_E + T_E f'^2_E 
		-R_E f'^1_{E\ix{1}}\Bigr\},\\[3pt]
D^{qq}_{11} &= \frac{1}{h}\Eint\; \zeta^{q,1}_E\Bigl\{
	\begin{aligned}[t]
	& \zeta^{q,1}_E( f'^1_E +  T_E f'^2_E - R_E f'^1_{E\ix{1}})\\
	& -2\hbar\w R_E f'^1_{E\ix{1}}\Bigr\},
	\end{aligned}\\[3pt]
D^{qq}_{22} &= \frac{1}{h}\Eint\; \Bigl\{
	\begin{aligned}[t]
	& (\hbar\omega)^2 R_E f'^2_E
	 -(\zeta^{q,2}_{E\ix{1}})^2 R_E f'^1_{E\ix{1}}\\
	& +(\zeta^{q,2}_E)^2(f'^1_E + T_E f'^2_E)\Bigr\}. 
	\end{aligned}
\end{align}
\end{subequations}
Finally, the shot-noise contributions to the current fluctuations, 
$R^{\r\r}_{\a\a}$ and $R^{qq}_{\a\a}= \sum\nolimits_{uv}(\d_{u\ve}-\mu_\a\d_{u\r})(\d_{v\ve}
-\mu_\a\d_{v\r})R^{uv}_{\a\a}$ are obtained from 
Eq.~\eqref{TwMENSh} as
\begin{subequations}
\begin{align}
R^{\r\r}_{11}&= R^{\r\r}_{22} = \frac{1}{h}\Eint \;
	 T_E R_E (f^1_{E\ix{1}}-f^2_E)^2,\\[3pt]
R^{qq}_{11} &= \frac{1}{h}\Eint \; 
	(\zeta^{q,1}_E)^2 T_E R_E (f^1_{E\ix{1}}-f^2_E)^2,\\[3pt]
R^{qq}_{22} &= \frac{1}{h}\Eint \; 
	(\zeta^{q,2}_{E\ix{1}})^2 T_E R_E  (f^1_{E\ix{1}}-f^2_E)^2.
\end{align}
\end{subequations}

In order to evaluate the coefficient \eqref{FFE_ER4} for $x=\r$ and 
$\a=1$, we further need the mean currents $\tilde{J}^x_\a$ as well as
the thermal and shot-noise fluctuations $\tilde{D}^{\r\r}_{11}$ and 
$\tilde{R}^{\r\r}_{11}$ for the time-reversed system. 
These quantities, which are found by replacing $S^{\a\b}_{\En,E}$ 
with $\mathsf{T}_{\mathbf{B}}\mathsf{T}_{\df}S^{\a\b}_{\En,E} =
S^{\b\a}_{E,\En}$ in Eqs.~\eqref{TwMeanCurrents} and \eqref{TwMEN}, 
are given by 
\begin{subequations}
\begin{align}
\tilde{J}^{\r}_1 &= -\tilde{J}^{\r}_2 = \frac{1}{h}\Eint\;\Bigl\{
	T_E f^1_E +R_E f^2_{E\ix{1}}- f^2_E\Bigr\},\\[3pt]
\tilde{J}^{q}_1 &= \frac{1}{h}\Eint\;\Bigl\{
	\zeta^{q,1}_E (f^1_E-f^2_E) 
	+ \zeta^{q,1}_{E\ix{1}} R_E (f^2_{E\ix{1}}-f^1_E)\Bigr\},\\
\tilde{J}^{q}_2 &= \frac{1}{h}\Eint \; \zeta^{q,2}_E\Bigl\{
	f^2_E - R_Ef^2_{E\ix{1}} -T_E f^1_E\Bigr\}
\end{align}
\end{subequations}
and 
\begin{subequations}
\begin{align}
\tilde{D}^{\r\r}_{11} &= \frac{1}{h}\Eint\;\Bigl\{
	f'^2_E+ T_Ef'^1_E -R_E f'^2_{E\ix{1}}\Bigr\},\\[3pt]
\tilde{R}^{\r\r}_{11} &= \frac{1}{h}\Eint \;
	T_E R_E (f^1_E -f^2_{E\ix{1}})^2.
\end{align}
\end{subequations}

\end{document}